\documentclass[10pt]{article}
\usepackage[margin=26mm]{geometry}
\usepackage{amsthm}
\usepackage{amsmath}
\usepackage{mathbbol}
\usepackage{amsfonts}
\usepackage{amssymb}
\usepackage{graphicx}
\usepackage{algorithm}
\usepackage{url}
\usepackage{epsfig}
\usepackage{wrapfig}
\usepackage{color}
\usepackage{blkarray}

\usepackage{arydshln}

\usepackage{subfig}

\usepackage{algorithm}
\usepackage{algpseudocode}
\usepackage{comment}

\usepackage{pgf,tikz}
\usepackage[all]{xy}
\usetikzlibrary{matrix, positioning}
\usetikzlibrary{patterns}
\usepackage{tikz-cd}
\usetikzlibrary{arrows}
\usepackage{pbox}
\usepackage[capitalise]{cleveref}
\usepackage{xspace}
\usepackage{xstring}
\usepackage{multicol}

\newcommand{\R}{\mathbb{R}}

\newcommand{\F}{\mathbb{F}}

\newcommand{\Sh}{\text{Sh}}

\newcommand{\eps}{\epsilon}
\newcommand{\free}{\ast}
\newcommand\scalemath[2]{\scalebox{#1}{\mbox{\ensuremath{\displaystyle #2}}}}

\newcommand{\dnoise}{d_{\mathrm{noise}}}

\newcommand{\ring}{\mathcal{R}}

\newcommand{\ignore}[1]{}

\newtheorem{theorem}{Theorem}
\newtheorem{lemma}[theorem]{Lemma}

\newtheorem{corollary}[theorem]{Corollary}

\DeclareMathOperator\coker{coker}

\title{Computing the interleaving distance is NP-hard}
\author{H{\aa}vard B. Bjerkevik, Magnus B. Botnan, Michael Kerber}
\date{\today}

\begin{document}
\maketitle

\begin{abstract}
We show that computing the interleaving distance between two multi-graded
persistence modules is NP-hard. More precisely, we show
that deciding whether two modules are $1$-interleaved is NP-complete,
already for bigraded, interval decomposable modules.
Our proof is based on previous work showing 
that a constrained matrix invertibility problem can be reduced to the
interleaving distance computation of a special type of persistence modules.
We show that this matrix invertibility problem is NP-complete.
We also give a slight improvement of the above reduction, showing that
also the approximation of the interleaving distance is NP-hard
for any approximation factor smaller than $3$. Additionally, we obtain corresponding hardness results for the case that the modules are indecomposable, 
and in the setting of one-sided stability. Furthermore, we show that checking for injections (resp. surjections) between persistence modules is NP-hard. In conjunction with earlier results from computational algebra this gives a complete characterization of the computational complexity of one-sided stability. Lastly, we show that it is in general NP-hard to approximate distances induced by noise systems within a factor of 2. 
\end{abstract}

\section{Introduction}
\paragraph{Motivation and problem statement.}
A \emph{persistence module} $M$ over $\R^d$ is a collection of vector spaces
$\{M_p\}_{p\in\R^d}$ and linear maps $M_{p\to q}\colon M_p\to M_q$
whenever $p\leq q$, with the property that
$M_{p\to p}$ is the identity map 
and the linear maps are composable in the obvious
way. For $d=1$, we will talk about \emph{single-parameter persistence},
and for $d\geq 2$, we will use the term
\emph{multi-parameter persistence}.

Persistence, particularly in its single-parameter version, 
has recently gained a lot of attention
in applied fields, because one of its instantiations is
\emph{persistent homology}, which studies 
the evolution of homology groups when varying a real scale parameter.
The observation that topological features
in real data sets carry important information
to analyze and reason about the contained data
has given rise to the term \emph{topological data analysis (TDA)}
for this research field, with various connections
to application areas, e.g.~\cite{rhbk-stable,cgos-persistence,bmmps-persistent,pewvkjw-topology,rybakken2017decoding}.

A recurring task in TDA is the comparison of two persistence modules.
The natural notion in terms of algebra is by interleavings
of two persistence modules: given two persistence modules 
$M$ and $N$ as above and some $\eps>0$, 
an \emph{$\eps$-interleaving} is the assignment of
maps $\phi_p\colon M_p\to N_{p+\eps}$ and $\psi_p\colon N_p\to M_{p+\eps}$
which commute with each other and the internal maps
of $M$ and $N$. The \emph{interleaving} distance
is then just the infimum over all $\eps$ for which an interleaving exists.

A desirable property for any distance on persistence modules is 
\emph{stability},
meaning informally that a small change in the input data set should only
lead to a small distortion of the distance. 
At the same time, we aim for a \emph{sensitive} measure, meaning
that the distance between modules should be generally as large as possible
without violating stability. As an extreme example, the distance measure
that assigns $0$ to all pairs of modules is maximally stable,
but also maximally insensitive. Lesnick~\cite{lesnick-theory} 
proved that among all
stable distances for single- or multi-parameter persistence,
the interleaving distance is the most sensitive one over prime fields.
This makes the interleaving distance an interesting measure to be 
used in applications, and raises the question of how costly
it is to compute the distance~\cite[Sec.~1.3 and 7]{lesnick-theory}.
Of course, for the sake of computation,
a suitable finiteness condition must be 
imposed on the modules to ensure that they can be represented in finite
form; we postpone the discussion to Section~\ref{sec:background},
and simply call such modules \emph{of finite type}.

The complexity of computing the interleaving distance is well understood
for the single parameter case. The \emph{isometry theorem}~\cite{chazal2016structure, lesnick-theory} states the equivalence of the interleaving distance and the \emph{bottleneck
distance}, which is defined in terms of the \emph{persistence diagrams}
of the persistence modules
and can be reduced to the computation of a min cost bottleneck matching
in a complete bipartite graph~\cite{eh-computational}. 
That matching, in turn,
can be computed in $O(n^{1.5}\log n)$ time, and efficient
implementations have been developed recently~\cite{kmn-geometry}.

The described strategy, however, fails in the multi-parameter case,
simply because the two distances do not match for more than one parameter:
even if the multi-parameter persistence module admits a decomposition
into \emph{intervals} (which are ``nice'' indecomposable elements,
see Section~\ref{sec:background}), it has been proved that the
interleaving distance and the multi-parameter extension of the bottleneck
distance are arbitrarily far from each other~\cite[Example 9.1]{bl-algebraic}. Another example where the interleaving and bottleneck distances differ is given in~\cite[Example 4.2]{bjerkevik-stability}; moreover, in this example the pair of persistence modules has the property that potential interleavings can be written on a particular matrix form, later formalized by the introduction of \textit{CI problems} in~\cite{bb-computational}. A consequence is that the strategy of computing interleaving distance by computing the bottleneck distance fails also in this special case.

\paragraph{Our contributions.}
We show that, for $d=2$, the computation of the interleaving distance
of two persistence modules of finite type is NP-hard,
even if the modules are assumed to be decomposable into intervals.
In~\cite{bb-computational},
it is proved that the problem is CI-hard, where CI is a combinatorial
problem related to the invertibility of a matrix with a prescribed
set of zero elements. This is done by associating a pair of modules to each CI problem such that the modules are $1$-interleaved if and only if the CI problem has a solution. We ``finish'' this proof by showing that CI is
NP-complete, hence proving the main result. The hardness result on CI
is independent of all topological concepts required for the rest of the paper
and potentially of independent interest in other algorithmic areas.

Moreover, we slightly improve the reduction from~\cite{bb-computational}
that asserts the CI-hardness of the interleaving distance, showing
that also obtaining a $(3-\eps)$-approximation of the interleaving distance
is NP-hard to obtain for every $\eps>0$.
This result follows from the fact that our improved construction 
takes an instance of a CI problem and returns
a pair of persistence modules which are $1$-interleaved if the instance
has a solution and are $3$-interleaved if no solution exists.
We mention that for \emph{rectangle decomposable} modules in $d=2$,
a subclass of interval decomposable modules,
it is known that the bottleneck distance $3$-approximates the interleaving
distance~\cite[Theorem 3.2]{bjerkevik-stability}, and can be computed in polynomial time.
While this result does not directly extend to all interval decomposable
modules, it gives reason to hope that a $3$-approximation of
the interleaving distance exists for a larger class of modules.

We also extend our hardness result to related problems:
we show that it is NP-complete to compute the interleaving distance
of two \emph{indecomposable} persistence modules (for $d=2$).
We obtain this result by ``stitching'' together the interval decomposables
from our main result into two indecomposable modules without affecting
their interleaving distance. We remark that the restriction of
computing the interleaving distance of indecomposable \textit{interval} modules
has recently been shown to be in P~\cite{dx-computing}.

Bauer and Lesnick~\cite{bauer2014induced} showed that the existence of an interleaving pair, for modules indexed over $\R$, is equivalent to the existence of a single morphism with kernel and cokernel of a corresponding ``size''. While the equivalence does not hold in general, the two concepts are still closely related for $d>1$. Using this, we obtain as a corollary to the aforementioned results that it is in general NP-complete to decide if there exists a morphism whose kernel and cokernel have size bounded by a given parameter. We also show that it is NP-complete to decide if there exists a surjection (dually, an injection) from one persistence module to another. Together with the result of~\cite{brooksbank2008testing}, this gives a complete characterization of the computational complexity of ``one-sided stability''. Furthermore, we remark that this gives an alternative proof of the fact that checking for injections (resp. surjections) between modules over a finite-dimensional algebra (over a finite field) is NP-hard. This was first shown in~\cite[Theorem 1.2]{ivanyos2010deterministic} (for arbitrary fields). The paper concludes with a result showing that it is in general NP-hard to approximate distances induced by noise systems (as introduced by Scolamiero et al. \cite {scolamiero2017multidimensional}) within a factor of 2. 

\paragraph{Outline.}
We begin with the hardness proof for CI in \cref{sec:CI}.
In \cref{sec:background}, we discuss the 
representation-theoretic concepts needed in the paper.
In \cref{sec:reduction}, we describe our improved
reduction scheme from interleaving distance to CI.
In \cref{sec:indecomposable_hardness}, we prove the hardness
for indecomposable modules.
In \cref{sec:one_sided}, we prove our hardness result for one-sided
stability. A result closely related to one-sided stability can be found in \ref{sec:noise} where we discuss a particular distance induced by a noise system. We conclude in \cref{sec:conclusion}.

\section{The CI problem}
\label{sec:CI}
Throughout the paper, we set $\F$ to be any finite field
with a constant number of elements.
We write $\F^{n\times n}$ for the set of $n\times n$-matrices
over $\F$, and $P_{ij}\in\F$ for the entry of $P$ at the position at row
$i$ and column $j$. 
We write $I_n$ for the $n\times n$-unit matrix.
The \emph{constrained invertibility problem}
asks for a solution of the equation $AB=I_n$, when certain entries
of $A$ and of $B$ are constrained to be zero. 
Formally, using the notation $[n]:=\{1,\ldots,n\}$, we define the language
\begin{align*}
CI := &\{(n,P, Q) \mid P\subseteq [n]\times[n] \wedge Q \subseteq [n]\times [n] \wedge
\exists A,B\in \F^{n\times n}:\\ &\left(\forall (i,j)\in P: A_{i,j}=0\wedge\forall (i,j)\in Q: B_{i,j}=0\wedge AB=I_n\right)\}.
\end{align*}

We can write CI instances in a more visual form, for instance writing
\[\left(\begin{array}{ccc} \free &\free& \free\\\free& 0 & \free \\ \free&\free& 0\end{array}\right)\left(\begin{array}{ccc}\free&\free&\free\\\free& \free & 0 \\ \free&0& \free\end{array}\right) =I_3\]
instead of 
$(3,\{(2,2),(3,3)\},\{(2,3),(3,2)\})$. Indeed, the CI problem asks whether in the above matrices, we can fill the $\free$-entries
with field elements to satisfy the equation.
In the above example, this is indeed possible, for instance by choosing
\[A=\left(\begin{array}{ccc} 1&1&1\\1& \textbf{0} & 1 \\ 1&1& \textbf{0}\end{array}\right)\quad B=\left(\begin{array}{ccc} -1&1&1\\1&-1& \textbf{0} \\ 1& \textbf{0}&-1\end{array}\right).\]
We sometimes also call $A$ and $B$ a \emph{satisfying assignment}.
In contrast, the instance
\[\left(\begin{array}{ccc} 0 &\free& 0\\\free& \free & \free \\ \free&\free& \free\end{array}\right)\left(\begin{array}{ccc}\free&\free&\free\\0& \free & \free \\ \free&\free& \free\end{array}\right) =I_3
\]
has no solution, because the $(1,1)$ entry of the product on the left
is always $0$, no matter what values are chosen.
Note that the existence of a solution also depends on the characteristic of
the base field. For an example, see Chapter 4, page 13 in~\cite{bb-computational}.

The CI problem is of interest to us, because we will see 
in Section~\ref{sec:reduction} that CI \emph{reduces} to
the problem of computing the interleaving distance, that is,
a polynomial time algorithm for computing the interleaving
distance will allow us to decide whether a triple $(n,P,Q)$ is in $CI$,
also in polynomial time. Although the definition of CI is rather elementary
and appears to be useful in different contexts, we are not aware of any
previous work studying this problem (apart from~\cite{bb-computational}).

It is clear that $CI$ is in NP because a valid choice of the matrices
$A$ and $B$ can be checked in polynomial time. We want to show that
$CI$ is NP-hard as well. It will be convenient to do so in two steps.
First, we define a slightly more general problem,
called \emph{generalized constrained invertibility (GCI)}, and show that
GCI reduces to CI. Then, we proceed by showing that 3SAT reduces
to GCI, proving the NP-hardness of CI.

\paragraph{Generalized constrained invertibility.}
We generalize from the above problem in two ways: first, instead
of square matrices, 
we allow that $A\in\F^{n\times m}$ and $B\in\F^{m\times n}$
(where $m$ is an additional input).
Second, instead of forcing $AB=I_n$, we only require
that $AB$ coincides with $I_n$ in a fixed subset of entries over $[n]\times [n]$.
Formally, we define
\begin{align*}
GCI:= &\{(n,m,P,Q,R)\mid P\subseteq [n]\times[m]\wedge Q\subseteq [m]\times [n]\wedge R\subseteq [n]\times[n]\wedge\exists A\in \F^{n\times m},B\in\F^{m\times n}:\\
&\left(\forall (i,j)\in P: A_{i,j}=0\wedge\forall (i,j)\in Q: B_{i,j}=0\wedge \forall (i,j)\in R:(AB)_{i,j}=(I_n)_{i,j}\right)\}.
\end{align*}

Again, we use the following notation
\[
\left(\begin{array}{ccc}\free&\free&\free\\ 0 & 0 & 0\end{array}\right)
\left(\begin{array}{cc}\free&0\\0& \free \\ 0&0 \end{array}\right) 
=
\left(\begin{array}{cc} 1 & 0 \\ \free & \free\end{array}\right)
\]
for the GCI instance $(2,3,\{(2,1),(2,2),(2,3)\},\{(1,2),(2,1),(3,1),(3,2)\},\{(1,1),(1,2)\})$. This instance is indeed in GCI, as for instance,
\[
\left(\begin{array}{ccc}1 &0 & 0\\\textbf{0} & \textbf{0} & \textbf{0}\end{array}\right)
\left(\begin{array}{cc} 1 &\textbf{0}\\\textbf{0}& 1 \\ \textbf{0}&\textbf{0} \end{array}\right) 
=
\left(\begin{array}{cc} \textbf{1} & \textbf{0} \\ 0 & 0\end{array}\right).
\]

GCI is indeed generalizing CI, 
as we can encode any CI instance by setting $m=n$
and $R=[n]\times[n]$. Hence, CI trivially reduces to GCI. 
We show, however, that also the converse is true, meaning that the
problems are computationally equivalent.
We will need the following lemma which follows from linear algebra:

\begin{lemma}
\label{lem:padding_lemma}
Let $M\in\F^{n\times m}$, $N\in\F^{m\times n}$ with $m>n$
such that $MN=I_n$. Then there exist
matrices $M'\in\F^{(m-n)\times m}$, $N'\in\F^{m\times (m-n)}$
such that 
\[\left[\begin{array}{c} M\\M'\end{array} \right]
\left[\begin{array}{cc} N &N'\end{array} \right] = I_m.\]
\end{lemma}
\begin{proof}
Pick $M''\in\F^{(m-n)\times m}$ so that $\left[\begin{array}{c} M\\M''\end{array} \right]$ has full rank. This is possible, as the row vectors of $M$ are linearly independent, so we can pick the rows in $M''$ iteratively such that they are linearly independent of each other and those in $M$. Let $M' = M''-M''NM$, which gives \[M'N = (M''-M''NM)N = M''N-M''NI_n=0_{(m-n)\times n}.\] Since
\[\left[\begin{array}{c} M\\M''\end{array} \right] = \left[\begin{array}{cc} I_n&0_{n\times (m-n)}\\M''N&I_{m-n}\end{array} \right]
\left[\begin{array}{c} M\\M'\end{array} \right],\] $\left[\begin{array}{c} M\\M'\end{array} \right]$ also has full rank, which means that it has an inverse. Let $N'$ be the last $m-n$ columns of this inverse matrix. We get
\[\left[\begin{array}{c} M\\M'\end{array} \right]
\left[\begin{array}{cc} N &N'\end{array} \right] = \left[\begin{array}{cc} I_n &0_{n\times (m-n)}\\0_{(m-n)\times n}&I_{m-n} \end{array} \right] = I_m.\]
\end{proof}

\begin{lemma}
\label{lem:cgi_to_ci}
GCI is polynomial-time reducible to CI.
\end{lemma}
\begin{proof}
Fix a GCI instance $(n,m,P,Q,R)$. We have to define a polynomial time algorithm
to compute a CI instance $(n',P',Q')$ such that
\[(n,m,P,Q,R)\in GCI \Leftrightarrow (n',P',Q')\in CI.\]

Write the GCI instance as $AB=C$, where $A$ and $B$
are matrices with $0$ and $\free$ entries (of dimensions $n\times m$
and $m\times n$, respectively), and $C$ is an $n\times n$-matrix
with $1$ or $\free$ entries on the diagonal, and $0$ or $\free$ entries
away from the diagonal (as in the example above).

Define the matrix $I_n^\free$ as the matrix with $0$ away from the diagonal
and $\free$ on the diagonal. Moreover, let $\bar{C}$ denote the matrix
$C$ with all $1$-entries replaced by $0$-entries. Now, consider the
GCI instance
\begin{eqnarray}
\left[\begin{array}{cc} A& I_n^\free\end{array}\right] \left[ \begin{array}{c} B\\\bar{C}\end{array}\right]=I_n
\label{eqn:gci_instance}
\end{eqnarray}
which can be formally written as $(n,n+m,P',Q',[n]\times [n])$
for some choices $P'\supseteq P$, $Q'\supseteq Q$.

We claim that the original instance is in GCI if and only if the extended instance is in GCI. First, assume that $AB=C$ has a solution (that is, an assignment
of field elements to $\free$ entries that satisfies the equation).
Then, we pick all diagonal entries in $I_n^\free$ as $1$, so that the matrix
becomes $I_n$. Also, we pick $\bar{C}$ to be $I_n-AB$; this is 
indeed possible, as an entry in $\bar{C}$ is fixed only if the corresponding positions of 
$I_n$ and $AB$ coincide. With these choices, we have that
\[\left[\begin{array}{cc} A& I_n^\free\end{array}\right]\left[ \begin{array}{c} B\\\bar{C}\end{array}\right] = AB+I_n(I_n-AB)=I_n,\]
as required.

Conversely, if there is a solution for the extended instance, write
$X$ for the assignment of $I_n^\free$ and $Y$ for the assignment of $\bar{C}$.
Then $AB+XY=I_n$. Now fix any index $(i,j)\in R$ and consider the equation
in that entry. By construction $Y_{i,j}=0$, and multiplication by 
the diagonal matrix $X$ does not change this property.
It follows that $(AB)_{(i,j)}=(I_n)_{i,j}$, which means that $AB=C$
has a solution. Hence, the two instances are indeed equivalent.

To finish the proof, we observe that (\ref{eqn:gci_instance})
is in GCI if and only if
\begin{eqnarray}
\left[\begin{array}{cc} A& I_n^\free\\ \free_{m\times m} & \free_{m\times n}\end{array}\right] \left[ \begin{array}{cc} B & \free_{m\times m}\\\bar{C}& \free_{n\times m}\end{array}\right]=I_{n+m}
\label{eqn:gci_padded}
\end{eqnarray}
is in GCI, where $\free_{a\times b}$ is simply the $a\times b$ matrix
only containing $\free$ entries. 
Formally written, this instance corresponds to
$(n+m,n+m,P',Q',[n+m]\times [n+m])$.
To see the equivalence, if (\ref{eqn:gci_instance}) is in GCI,
Lemma~\ref{lem:padding_lemma} asserts that there are indeed choices
for the $\free$-matrices to solve (\ref{eqn:gci_padded}) as well.
In the opposite direction, a satisfying assignment of the involved matrices
in (\ref{eqn:gci_padded})
also yields a valid solution for (\ref{eqn:gci_instance}) when restricted to the upper $n$ rows and left $n$ columns, respectively.

Combining everything, we see that $(n,m,P,Q,R)$ is in GCI
if and only if $(n+m,n+m,P',Q',[n+m]\times [n+m])$ is in GCI.
The latter, however, is equivalent to the CI instance $(n+m,P',Q')$.
The conversion can clearly be performed in polynomial time, and the statement
follows.
\end{proof}

\paragraph{Hardness of GCI.}
We describe now how an algorithm for how deciding GCI can be used to
decide satisfiability of 3SAT formulas. Let $\phi$ be a 3CNF formula
with $n$ variables and $m$ clauses. We construct a GCI instance that
is satisfiable if and only if $\phi$ is satisfiable.

In what follows, we will often label some $\free$ entries in matrices
with variables when we want to talk about the possible assignments
of the corresponding entries.

The first step is to build a ``gadget'' that allows us to encode
the truth value of a variable in the matrix. Consider the instance

\[
\left(
\begin{array}{ccc}
\free&0&\free\\
0&\free&\free\\
\end{array}
\right)
\left(
\begin{array}{cc}
x&0\\
0&y\\
\free&\free\\
\end{array}
\right)
=
I_2.
\] 
In any solution to this equation, not both $x$ and $y$ can be zero
because otherwise, the right matrix would have rank at most $1$. Furthermore,
when extending the instance by one row/column
\[
\left(
\begin{array}{ccc}
a&b&0\\\hdashline
\free&0&\free\\
0&\free&\free\\
\end{array}
\right)
\left(
\begin{array}{c:cc}
0&x&0\\
0&0&y\\
0&\free&\free\\
\end{array}
\right)
=
\left(
\begin{array}{ccc}
\free & 0 & 0\\
0 & 1 & 0\\
0 & 0 & 1
\end{array}
\right),
\] 
we see that both $ax=0$ and $by=0$ must hold, which is then only possible
if at least one entry $a$ or $b$ is equal to $0$. In fact, there is a solution
with $a\neq 0$, and a solution with $b\neq 0$, for instance

\[
\left(
\begin{array}{ccc}
1&0&\textbf{0}\\
0&\textbf{0}&1\\
\textbf{0}&1&0
\end{array}
\right)
\left(
\begin{array}{ccc}
\textbf{0}&0&\textbf{0}\\
\textbf{0}&\textbf{0}&1\\
\textbf{0}&1&0\\
\end{array}
\right)
=
\left(
\begin{array}{ccc}
0 & \textbf{0} & \textbf{0}\\
\textbf{0} & \textbf{1} & \textbf{0}\\
\textbf{0} & \textbf{0} & \textbf{1}
\end{array}
\right),
\] 

\[
\left(
\begin{array}{ccc}
0&1&\textbf{0}\\
1&\textbf{0}&0\\
\textbf{0}&0&1
\end{array}
\right)
\left(
\begin{array}{ccc}
\textbf{0}&1&\textbf{0}\\
\textbf{0}&\textbf{0}&0\\
\textbf{0}&0&1\\
\end{array}
\right)
=
\left(
\begin{array}{ccc}
0 & \textbf{0} & \textbf{0}\\
\textbf{0} & \textbf{1} & \textbf{0}\\
\textbf{0} & \textbf{0} & \textbf{1}
\end{array}
\right).
\] 

The intuition is that for a variable $x_i$ appearing in $\phi$, 
we interpret $x_i$ to be true if $a\neq 0$, and to be false if $b\neq 0$.
We build such a gadget for each variable. 
A crucial observation is that we can do so with all variable entries
placed in the same row. This works essentially by concatenating the
variable gadgets, in a block-like fashion. We show the construction
for three variables as an example.

\[
\left(
\begin{array}{cccccccccc}
a_0&b_0&0&a_1&b_1&0&a_2&b_2&0&\free\\
\free&0&\free&0&0&0&0&0&0&0\\
0&\free&\free&0&0&0&0&0&0&0\\
0&0&0&\free&0&\free&0&0&0&0\\
0&0&0&0&\free&\free&0&0&0&0\\
0&0&0&0&0&0&\free&0&\free&0\\
0&0&0&0&0&0&0&\free&\free&0\\
\end{array}
\right)
\left(
\begin{array}{ccccccc}
0&\free&0&0&0&0&0\\
0&0&\free&0&0&0&0\\
0&\free&\free&0&0&0&0\\
0&0&0&\free&0&0&0\\
0&0&0&0&\free&0&0\\
0&0&0&\free&\free&0&0\\
0&0&0&0&0&\free&0\\
0&0&0&0&0&0&\free\\
0&0&0&0&0&\free&\free\\
\free&0&0&0&0&0&0\\
\end{array}
\right)
=
I_7
\]
where we introduced an additional column at the end of the left matrix
and an additional row at the end of the second matrix.
Firstly, this allows us to satisfy the entire $I_7$
on the right hand side; moreover, it will be useful when extending
the construction to clauses.
It is straightforward to generalize this construction to an arbitrary
number of variables. We arrive at the following intermediate result.

\begin{lemma}
\label{lem:variable_lemma}
For any $n\geq 1$, there exists a GCI instance $A'B'=I_{2n+1}$
with $A'$ having $3n+1$ columns, such that in each solution for the 
problem, $A'_{1,3n+1}$ is not zero, and for each $k=0,\ldots,n-1$, the entries
$A'_{1,3k+1}$ and $A'_{1,3k+2}$ are not both non-zero. Moreover, for any choice
of $v_1,\ldots,v_n\in\{1,2\}$, there exists a solution of the instance
in which $A'_{1,3k+v_i}\neq 0$ for all $k=0,\ldots,n-1$.
\end{lemma}

Next, we extend the instance from Lemma~\ref{lem:variable_lemma}
with respect to the clauses. 
We refer to the clauses as $c_1,\ldots,c_m$.
For each clause we append one further row to $A'$, 
each of them identical of the form 
\[\left(\begin{array}{ccccc} 0 & \ldots & 0 & \free\end{array}\right)^T.\]
We also append one column to $B'$ for each clause, each of length $3n+1$.
For each clause, the entry at row $3n+1$ is set to $\free$.
If a clause contains a literal of the form $x_i$ (in positive form),
we set the entry at row $3i+1$ to $\free$. If it contains a literal $\neg x_i$,
we set the entry at row $3i+2$ to $\free$. In this way, at most $4$ entries
in the column are fixed to $\free$, and we fix all other entries to be $0$.
Continuing the above example, 
for the clause $x_0\vee\neg x_1\vee x_2$, we obtain
a column of the form
\[
\left(
\begin{array}{cccccccccc}
\free&0&0&0&\free&0&\free&0&0&\free\\
\end{array}
\right)
\]
Let $A$ and $B$ denote the matrices extended from $A'$ and $B'$ with the
above procedure. We next define $C$ as a square matrix of dimension $2n+1+m$
as follows: The upper left $(2n+1)\times (2n+1)$ submatrix is set to
$I_{2n+1}$. The rest of the first row is set to $0$, and the rest of
the diagonal is set to $1$. All other entries are set to $\free$. 
This concludes the description of a GCI instance $AB=C$
out of a 3CNF formula $\phi$.
We exemplify the construction for the formula  
$(x_0 \vee x_1 \vee \neg x_2) \wedge (\neg x_0 \vee x_1 \vee x_2)$,
where the lines mark the boundary of $A'$ and $B'$, respectively.
\[
\scalemath{0.9}{
\left(
\begin{array}{cccccccccc}
a_0&b_0&0&a_1&b_1&0&a_2&b_2&0&\free\\
\free&0&\free&0&0&0&0&0&0&0\\
0&\free&\free&0&0&0&0&0&0&0\\
0&0&0&\free&0&\free&0&0&0&0\\
0&0&0&0&\free&\free&0&0&0&0\\
0&0&0&0&0&0&\free&0&\free&0\\
0&0&0&0&0&0&0&\free&\free&0\\
\hline
0&0&0&0&0&0&0&0&0&\free\\
0&0&0&0&0&0&0&0&0&\free\\
\end{array}
\right)
\left(
\begin{array}{ccccccc|cc}
0&\free&0&0&0&0&0&\free&0\\
0&0&\free&0&0&0&0&0&\free\\
0&\free&\free&0&0&0&0&0&0\\
0&0&0&\free&0&0&0&\free&\free\\
0&0&0&0&\free&0&0&0&0\\
0&0&0&\free&\free&0&0&0&0\\
0&0&0&0&0&\free&0&0&\free\\
0&0&0&0&0&0&\free&\free&0\\
0&0&0&0&0&\free&\free&0&0\\
\free&0&0&0&0&0&0&\free&\free\\
\end{array}
\right)
=
\left(
\begin{array}{ccccccccc}
1&0&0&0&0&0&0&0&0\\
0&1&0&0&0&0&0&\free&\free\\
0&0&1&0&0&0&0&\free&\free\\
0&0&0&1&0&0&0&\free&\free\\
0&0&0&0&1&0&0&\free&\free\\
0&0&0&0&0&1&0&\free&\free\\
0&0&0&0&0&0&1&\free&\free\\
\free&\free&\free&\free&\free&\free&\free&1&\free\\
\free&\free&\free&\free&\free&\free&\free&\free&1\\
\end{array}
\right)
}
\]

\begin{lemma}
\label{lem:clause_lemma}
$AB=C$ admits a solution if and only if $\phi$ is satisfiable.
\end{lemma}
\begin{proof}
``$\Rightarrow$'': Let us assume that $AB=C$ has a solution,
which also implies a solution $A'B'=I_{2n+1}$
being a subproblem encoded in the instance.
Fixing a solution, we assign an assignment of the variables of $\phi$
as follows: If the entry $A_{1,3i+1}$ is non-zero, we set 
$x_i$ to true. If the entry $A_{1,3i+2}$ is non-zero,
we set $x_i$ to false. If neither is non-zero, we set $x_i$ to false as well
(the choice is irrelevant). Note that by Lemma~\ref{lem:variable_lemma},
not both $A_{1,3i+1}$ and $A_{1,3i+2}$ can be non-zero, so the assignment is
well-defined.

First of all, let $\gamma$ be the rightmost entry of the
first row of $A$. Because the $(1,1)$-entry of $C$ is set to $1$,
it follows that $\gamma\delta=1$, where $\delta$ is the lowest entry of
the first column of $B$. Hence, in the assumed solution, $\gamma\neq 0$.

Now fix a clause $c$ in $\phi$ and let $v$ denote the column of $B$
assigned to this clause, with column index $i$.
Recall that $v$ consists of (up to) three $\free$ entries chosen
according to the literals of $c$, and a $\free$ entry at the lowest position. Let $\lambda$ denote the value of that lowest entry in the assumed solution
of $AB=C$. We see that $\lambda\neq 0$, with a similar argument 
as for $\gamma$ above, using the $(i,i)$-entry of $C$.

Now, the $(1,i)$ entry of $C$ is set to $0$ by construction which yields
a constraint of the form
\[\mu_1 v_1 + \mu_2 v_2 + \mu_3 v_3 + \underbrace{\gamma\lambda}_{\neq 0} = 0\]
where $v_1$, $v_2$, $v_3$ are entries of $v$ at the $\free$ positions,
and $\mu_1$, $\mu_2$, $\mu_3$ 
the corresponding entries of the first row of $A$.
We observe that at least one term $\mu_jv_j$ must be non-zero,
hence both entries are non-zero.

This implies that the chosen assignment satisfies the clause: 
if $v_j$ is at index $3k+1$ for some $k$, the clause contains the literal
$x_k$ by construction and since $\mu_j\neq 0$, our assignment sets $x_k$
to true. The same argument applies to $v_j$ of the form $3k+2$.
It follows that the assignment satisfies all clauses and hence,
$\phi$ is satisfiable.

``$\Leftarrow$'': We pick a satisfying assignment for $\phi$
and fill the first row of $A$ as follows: if $x_i$ is true, 
we set $(A_{1,3i+1},A_{1,3i+2})$ to $(1,0)$ if $x_i$ is false,
we set it it $(0,1)$. By Lemma~\ref{lem:variable_lemma},
there exists a solution for $A'B'=I_{2n+1}$ with this initial
values and we choose such a solution, filling the upper 
$(2n+1)$ rows of $A$ and the left $(2n+1)$ columns of $B$.
Note that similar as above,
the value $\gamma$ at $A_{1,3n+1}$ must be non-zero in such a solution.
In the remaining $m$ rows of $A$, by construction, we only need to 
pick the rightmost entry, and we set it to $\gamma$ in each of these rows.
That determines all entries of $A$.

To complete $B'$ to $B$, we need to fix values in the 
columns of $B$ associated to clauses. In each such column, we pick 
the lowest entry to be $\frac{1}{\gamma}$, satisfying the constraints
of $C$ along the diagonal. Fixing a column $i$ of $B$, the $(1,i)$-constraint
of $C$ reads as
\[\mu_1 v_1 + \mu_2 v_2 + \mu_3 v_3 + \underbrace{\gamma\frac{1}{\gamma}}_{=1} = 0,\]
where $v_1,v_2,v_3$ are the remaining non-zero entries in $i$-th column.
Because we encoded a satisfying assignment of $\phi$ in the first row
of $A$, at least one $\mu_j$ entry is $1$.
We set the corresponding entry $v_j$ to $-1$, and the remaining $v_k$'s to $0$.
In this way, all constraints are satisfied,
and the GCI instance has a solution.
\end{proof}

Clearly, the GCI instance of the preceding proof can be computed
from $\phi$ in polynomial time.
It follows:

\begin{theorem}
\label{thm:ci_np_complete}
CI is NP-complete.
\end{theorem}
\begin{proof}
Lemma~\ref{lem:clause_lemma} shows the reduction of 3SAT to GCI,
proving that GCI is NP-complete.
As shown in Lemma~\ref{lem:cgi_to_ci}, GCI reduces to CI,
proving the claim.
\end{proof}

\section{Modules and Interleavings}
\label{sec:background}

In what follows, all vector spaces are understood
to be $\F$-vector spaces for the fixed base field $\F$.
Also, for points $p=(p_x,p_y),q=(q_x,q_y)$ in $\R^2$, 
we write $p\leq q$
if $p_x\leq q_x$ and $p_y\leq q_y$.

\paragraph{Persistence modules.}
A \emph{(two-parameter) persistence module} $M$
is a collection of $\F$-vector spaces $V_{p}$,
indexed over $p\in\R^2$ together with linear maps
$M_{p\to q}$ whenever $p\leq q$.
These maps must have the property that
$M_{p\to p}$ is the identity map on $M_p$
and $M_{q\to r}\circ M_{p\to q} = M_{p\to r}$
for $p\leq q\leq r$. Much more succinctly, a
persistence module is a functor from the poset category
$\R^2$ to the category of vector spaces. A \emph{morphism} between $M$ and $N$ is a collection of linear maps $\{f_p\colon M_p \to N_p\}$ such that $N_{p\to q} \circ f_p = f_q \circ M_{p\to q}$. We say that $f$ is an \emph{isomorphism} if $f_p$ is an isomorphism for all $p$, and denote this by $M\cong N$. If we view persistence modules as functors, a morphism is simply a natural transformation between the functors.

The simplest example is the \emph{$0$-module}
where $M_p$ is the trivial vector space for all $p\in\R^2$.
For a more interesting example, define
an \emph{interval} in the poset $(\R^2,\leq)$ to be a non-empty subset
$S\subset\R^2$ such that whenever $a,c\in S$ and $a\leq b\leq c$, then $b\in S$,
and moreover, if $a,c\in S$, there exists a sequence of elements $a=b_1,\ldots,b_\ell=c$
of elements in $S$ such that $b_i\leq b_{i+1}$ or $b_{i+1}\leq b_i$.
We associate an \emph{interval module} $I^S$ to $S$ as follows: for $p\in S$, 
we set $I^S_p:=\F$, and $I^S_p:=0$ otherwise. As map $I^S_{p\to q}$ with $p\leq q$, 
we attach the identity map if $p,q\in S$, and the $0$-map otherwise.

For $a\in\R^2$, let $\langle a\rangle:=\{x\in\R^2\mid a\leq x\}$ 
be the infinite
rectangle with $a$ as lower-left corner.
Given $k$ elements $a_1,\ldots,a_k\in\R^2$, the set
\[S:=\bigcup_{i=1,\ldots,k} \langle a_i\rangle\]
is called the \emph{staircase} with elements $a_1,\ldots,a_k$.
We call $k$ the size of the staircase.
See Figure~\ref{fig:staircase_illu} for an illustration.
It is easy
to verify that $S$ is an interval for $k\geq 1$. Clearly, if $a_i\leq a_j$,
we can remove $a_j$ without changing the staircase, so we assume that 
the elements forming the staircase are pairwise incomparable.
The \emph{staircase module} is the interval module associated to the staircase.

\begin{figure}
\centering
\includegraphics[width=5cm]{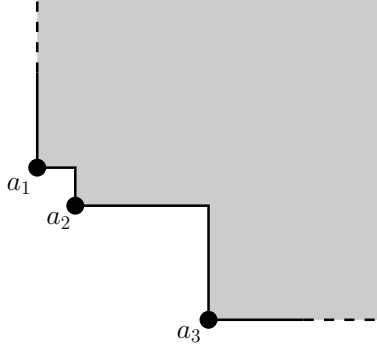}
\caption{A staircase of size $3$ (shaded area).}
\label{fig:staircase_illu}
\end{figure}

Given two persistence modules $M$ and $N$, the \emph{direct sum} $M\oplus N$
is the persistence module where $(M\oplus N)_p:=M_p\oplus N_p$,
and the linear maps are defined componentwise in the obvious way.
We call a persistence module $M$ \emph{indecomposable},
if in any decomposition $M=M_1\oplus M_2$, $M_1$ or $M_2$ is the $0$-module.
For example, it is not difficult to see that interval modules are indecomposable.
We call $M$ \emph{interval decomposable} if $M$ admits a decomposition
$M\cong M_1\oplus\ldots\oplus M_\ell$ into (finitely many) interval modules. The decomposition of any persistence module into interval modules is unique up to rearrangement and isomorphism of the summands; see \cite[Section 2.1]{bl-algebraic} and the references therein. This implies that there is a well-defined multiset of intervals $B(M)$ given by the decomposition of $M$ into interval modules. The multiset $B(M)$ is called the \emph{barcode} of $M$.
Not every module is interval decomposable; we remark that already rather
simple geometric constructions can give rise to complicated indecomposable elements \cite{be-realizations}.

\paragraph{Interleavings.}

Let $\eps \in \R$. For a persistence module $M$, the \emph{$\eps$-shift of $M$} is the module $M^\eps$ defined by $M_p^\eps = M_{p+\eps}$ (where $p+\eps=(p_x+\eps,p_y+\eps)$) and $M^\eps_{p \to q} = M_{p+\eps\to q+\eps}$. Note that $(M^\eps)^\delta = M^{\eps+\delta}$. As an example, staircase modules are closed under shift: the $\eps$-shift of the staircase module associated to $\bigcup \langle a_i\rangle$ is the staircase module associated to $\bigcup \langle a_i-\eps\rangle$. We can also define shift on morphisms: for $f\colon M \to N$, $f^\eps\colon M^\eps \to N^\eps$ is given by $f_p^\eps = f_{p+\eps}$. For $\eps\geq0$, there is an obvious morphism $\Sh_M(\eps)\colon  M \to M^\eps$ given by the internal morphisms of $M$, that is, we have $\Sh_M(\eps)_p = M_{p \to p+\eps}$. In practice we will often suppress notation and simply write $M \to M^\eps$ for this morphism.

With this in mind, we define an \emph{$\eps$-interleaving} between $M$ and $N$ for $\eps\geq 0$ as a pair $(f,g)$ of morphisms $f\colon  M \to N^\eps$ and $g\colon N \to M^\eps$ such that $g^\eps \circ f = \Sh_M(2\eps)$ and $f^\eps \circ g = \Sh_N(2\eps)$. Concretely, an $\eps$-interleaving between two persistence modules
$M$ and $N$ is a collection of maps
\begin{align}
f_p\colon M_p\to N_{p+\eps}\\
g_p\colon N_p\to M_{p+\eps}
\end{align}
such that all diagrams that can be composed out of
the maps $f_\ast$, $g_\ast$, and the linear maps of $M$ and $N$ commute. Note that a $0$-interleaving simply means that the persistence modules are \emph{isomorphic}. 
Also, an $\eps$-interleaving induces a $\delta$-interleaving for $\eps<\delta$ directly
by a suitable composition with the linear maps of the modules.

We say that two modules are \emph{$\eps$-interleaved} if there exists an $\eps$-interleaving
between them.
We define the \emph{interleaving distance} of two modules $M$ and $N$ as
\[d_I(M,N):=\inf\{\eps\geq 0\mid\text{$M$ and $N$ are $\eps$-interleaved}\}.\]
Note that $d_I$ defines an extended pseudometric on the space of persistence modules. The distance between two modules might be infinite, and there are non-isomorphic modules with distance $0$.
The triangle inequality follows from the simple observation that an $\eps$-interleaving
between $M_1$ and $M_2$ and a $\delta$-interleaving between $M_2$ and $M_3$ can be composed
to an $(\eps+\delta)$-interleaving between $M_1$ and $M_3$.

\paragraph{Representation of persistence modules.}
For studying the computational complexity of the interleaving distance,
we need to specify a finite representation of persistence modules
that allows us to pass such modules as an input to an algorithm.

A \emph{graded matrix representation} of a module $M$ 
is a $3$-tuple $(G,R,A)$,
where $G=\{g_1,\ldots,g_n\}$ is a list of $n$ points in $\R^2$,
$R=\{r_1,\ldots,r_m\}$ is a list of $m$ points in $\R^2$,
(with repetitions allowed),
and $A$ is an $(m\times n)$-matrix over the base field $\F$.
Equivalently, we can simply think of a matrix $A$ where each row and
column is annotated with a grading in $\R^2$.

The algebraic explanation for this representation is as follows:
it is known that a persistence module $M$ over $\R^2$ can be equivalently
described as a graded $\ring$-module over a suitably chosen ring $\ring$.
Assuming that $M$ is finitely presented, we can consider the free
resolution of $M$
\[\ring^m\stackrel{\partial^T}{\to}\ring^n\to M\to 0.\]
A graded matrix representation is simply a way to encode the map $\partial$
in this resolution. 

Let us describe for concreteness 
how a representation $(G,R,A)$ gives rise to a persistence module.
First, let $\F_1,\ldots,\F_n$ be copies of $\F$, and let
$e_i$ be the $1$-element of $\F_i$. For $p\in\R^2$, we define
$\mathrm{Gen}_p$ as the direct sum of all $\F_i$ such that
$g_i\leq p$. Moreover, every row of $A$ gives rise to a linear combination
of the entries $e_1,\ldots,e_n$. Let $c_i$ denote the linear combination
in row $i$. We define $\mathrm{Rel}_p$ to be the span
of all linear combinations $c_i$ for which $r_i\leq p$. Then, we set
\[M_p:=\frac{\mathrm{Gen}_p}{\mathrm{Rel}_p}\]
which is a $\F$-vector space. For $p\leq q$, 
writing $[x]_p$ for an element of $M_p$ with $x\in\mathrm{Gen}_p$, we define
\[M_{p\to q}([x]_p):=[x]_q.\]
It is easy to check that that $[x]_q$ is well-defined
(since $\mathrm{Gen}_p\subseteq\mathrm{Gen}_q$)
and independent of the chosen representative in $\mathrm{Gen}_p$
(since $\mathrm{Rel}_p\subseteq\mathrm{Rel}_q$).
Moreover, it is straightforward to verify that these maps satisfy
the properties of a persistence module.

\smallskip

In short, every persistence module that can be expressed by finitely
many generators and relations can be brought into graded matrix
representation. For instance, a staircase module for $a_1,\ldots,a_n$
of size $n$ where the $a_i$ are ordered by increasing first coordinate can be represented by a matrix with $n$ columns
graded by $a_1,\ldots,a_n$, and $n-1$ rows, where every row
corresponds to a pair $(i,i+1)$ with $1\leq i\leq n-1$.
In this row, we encode the relation $e_i=e_{i+1}$
and grade it by $p_{ij}$, which is the (unique) minimal element $q$
in $\R^2$ such that $a_i\leq q$ and $a_{i+1}\leq q$.
Hence, the graded matrix representation of a staircase of size $n$
has a size that is polynomial in $n$.

We also remark that a graded matrix representation is equivalent
to \emph{free implicit representations} \cite[Sec 5.1]{lw-interactive} 
for  the special case of $m_0=0$.

\section{Hardness of interleaving distance}
\label{sec:reduction}

We consider the following computational problems:

\begin{quote}
\textsc{$1$-Interleaving}: Given two persistence modules $M$, $N$ in 
graded matrix representation, decide whether they are $1$-interleaved.
\end{quote}

\begin{quote}
\textsc{$c$-Approx-Interleaving-Distance}: 
Given two persistence modules $M$, $N$ in 
graded matrix representation, return a real number $r$ such that
\[d_I(M,N)\leq r \leq c\cdot d_I(M,N)\]
\end{quote}

Obviously, the problem of computing $d_I(M,N)$ exactly is 
equivalent to the above definition with $c=1$.

The main result of this section is the following theorem:

\begin{theorem}
\label{1-3}
Given a CI-instance $(n,P,Q)$, we can compute in polynomial time in $n$
a pair of persistence modules $(M,N)$ in graded matrix representation
such that
\[
d_I(M,N)=
\begin{cases}
1 & \text{ if $(n,P,Q)\in CI$}\\
3 & \text{ if $(n,P,Q)\notin CI$}
\end{cases}.
\]
Moreover, both $M$ and $N$ are direct sums of staircase modules
and hence interval decomposable.
\end{theorem}

We will postpone the proof of Theorem~\ref{1-3} to the end of the section
and first discuss its consequences.

\begin{theorem}
\label{thm:1_interleaving_np_complete}
\textsc{$1$-Interleaving} is NP-complete.
\end{theorem}
\begin{proof}

We first argue that \textsc{$1$-Interleaving} is in NP.
First, note that to specify a $1$-interleaving, it suffices to specify
the maps at the points in $S$, where $S$ is a finite set whose size is polynomial in the size
of the graded matrix representation. More precisely, $S$ contains
the critical grades of the two modules
(that is, the grades specified by $G$ and $R$), as well
as the least common successors of such elements. That ensures that
every vector space (in both modules) can be isomorphically pulled back
to one of the elements of $S$, and the interleaving map can be defined using this
pull-back. It is enough to consider the points in $S$ to check if this set of pointwise maps is a valid morphism.

We can furthermore argue that verifying that a pair of such maps
yields a $1$-interleaving can be checked in a polynomial number of
steps. Again, this involves mostly the maps specified above,
as well as the corresponding maps shifted by $(1,1)$, in order to check
the compatibility of the two interleaving maps.
We omit further details of this step.

Finally, \textsc{$1$-Interleaving} is NP-hard: 
Assuming a polynomial time algorithm $A$ to decide the problem,
we can design a polynomial time algorithm for CI just by transforming $(n,P,Q)$
into a pair of modules $(M,N)$ using the algorithm from Theorem~\ref{1-3}.
If $A$ applied on $(M,N)$ returns true, we return that $(n,P,Q)$ is in CI.
Otherwise, we return that $(n,P,Q)$ is not in CI. 
Correctness follows from Theorem~\ref{1-3}, and the algorithm runs in
polynomial time, establishing a polynomial time reduction.
By Theorem~\ref{thm:ci_np_complete}, CI is NP-hard, hence,
so is \textsc{$1$-Interleaving}.
\end{proof}

\begin{theorem}
\label{thm:c_approx_interleaving_is_hard}
\textsc{$c$-Approx-Interleaving-Distance} is NP-hard for every $c<3$
(i.e., a polynomial time algorithm for the problem implies P=NP).
\end{theorem}
\begin{proof}
Fixing $c<3$, 
assuming a polynomial time algorithm $A$ for
\textsc{$c$-Approx-Interleaving-Distance}
yields a polynomial time algorithm for CI:
Given the input $(n,P,Q)$, we transform it into $(M,N)$
with Theorem~\ref{1-3}. Then, we apply $A$ on $(M,N)$. If the result is
less than $3$, we return that $(n,P,Q)$ is in CI.
Otherwise, we return that $(n,P,Q)$ is not in CI.
Correctness follows from Theorem~\ref{1-3}, noting that if $(n,P,Q)$ is
in CI, algorithm $A$ must return a number in the interval $[1,c]$
and $c<3$ is assumed. If $(n,P,Q)$ is not in CI, it returns a number $\geq 3$.
Also, the algorithm runs in polynomial time in $n$. Therefore, 
the existence of $A$ yields a polynomial time algorithm for CI,
implying P=NP with Theorem~\ref{thm:ci_np_complete}.
\end{proof}

Since the modules in Theorem~\ref{1-3} are direct sums of staircases,
both Theorem~\ref{thm:1_interleaving_np_complete} 
and Theorem~\ref{thm:c_approx_interleaving_is_hard} hold already for
the restricted case that the modules are interval decomposable.

\paragraph{Interleavings of staircases.}
The persistence modules constructed for the proof of Theorem~\ref{1-3}
will be direct sums of staircases. Before defining them, we 
establish some properties of the interleaving map between staircases 
and their direct sums which reveal the connection to the CI problems.

Recall from \cref{sec:background} that a morphism $M \to N$ can be described more concretely as a collection of maps $M_p\to N_p$ that are compatible with the linear maps in $M$ and $N$, that $M^\eps$ is defined by $M_p^\eps = M_{p+\eps}$, and that an $\eps$-interleaving is a pair of morphisms $\phi\colon  M \to N^\eps$, $\psi\colon  N \to M^\eps$ satisfying certain conditions. For staircase modules, the set of morphisms is quite limited.

For $M$ and $N$ staircase modules and $\lambda\in\F$, 
we denote by $1\mapsto\lambda$ the collection of linear maps $\phi_p$
such that $\phi_p(1)=\lambda$ for all $p$ such that $M_p=\F$.

\begin{lemma}
\label{lem:1_to_lambda}
Let $M$ and $N$ be staircase modules. Every morphism from $M$ to $N$ is of the form 
$1\mapsto\lambda$ for some $\lambda\in\F$.
\end{lemma}
\begin{proof}
Assume first that $p\leq q$ and $M_p=\F$. Write $\lambda:=\phi_p(1)$. 
Then, also $M_q=\F$, and $\phi_q(1)=\lambda$ as well, since
the linear maps from $p$ to $q$ for $M$ and $N$ are injective maps.

For incomparable $p$ and $q\in\R^2$, we consider the least common
successor $r$ of $p$ and $q$. Using the above property twice, we
see at once that $\phi_p(1)=\phi_r(1)=\phi_q(1)$.
\end{proof}

We examine next which values of $\lambda$
are possible for a concrete pair of staircases.
For a staircase $S$, let $S^\eps$ denote the staircase where each point
is shifted by $(\eps,\eps)$. This way, if $M$ is the module associated to $S$, $M^\eps$ is the module associated to $S^\eps$. As we noted before, the shift of a staircase module is also a staircase module.
Define the \emph{directed shift distance} from the staircase $S$ 
to the staircase $T$ as
\[d_s(S,T):=\min\{\eps\geq 0 \mid S\subseteq T^\eps\}.\]
One can show that the set on the right-hand side has a minimum value by using the fact that a staircase is generated by a finite set of elements, so $d_s$ is in fact well defined. Clearly, $d_s(S,T)\neq d_s(T,S)$ in general. The following simple observation
is crucial for our arguments. Let $M$, $N$ denote the staircase modules
induced by $S$ and $T$.

\begin{lemma}
\label{lem:transformation_lemma}
If $\eps<d_s(S,T)$, the only morphism from $M$ to $N^\eps$
is $1\mapsto 0$. If $\eps\geq d_s(S,T)$, every choice of $\lambda\in \F$
yields a morphism $1\mapsto\lambda$ from $M$ to $N^\eps$.
\end{lemma}
\begin{proof}
In the first case, by construction, there exists
some $p$ such that $M_p=\F$, but $N_{p+\eps}=0$. 
Hence, $0$ is the only choice for $\lambda$.

In the second case, $M_p=\F$ implies $N_{p+\eps}=\F$
as well. It is easy to check that any choice of $\lambda$
yields a compatible collection of maps,
hence a morphism.
\end{proof}

In particular, there are morphisms $M \to N$ given by arbitrary elements of $\F$ if and only if $S\subseteq T$. As a consequence, we can characterize morphisms of direct sums of staircase modules.

\begin{lemma}
\label{lem:matrix_rep_lemma}
Let $M=\oplus_{i=1}^n M_i$ and $N=\oplus_{j=1}^n N_j$ be direct sums of staircase modules. Then a collection of maps $\phi_p\colon M_p \to N_p$ is a morphism if and only if the restriction to $M_i$ and $N_j$ is a morphism
for any $i,j\in\{1,\ldots,n\}$. 
Therefore, a morphism $\phi$
is determined by an $(n\times n)$-matrix with entries in $\F$.
\end{lemma}
\begin{proof}
Let $p\leq q$, and consider the following diagram:
\[
\begin{tikzcd}
M_p\ar[d, "\phi_p"]\ar[rr,"M_{p\to q}"] & &  M_q\ar[d, "\phi_q"] \\
N_{p}\ar[rr,swap, "N_{p \to q}"] & & N_{q}
\end{tikzcd}
\]
We have $M_p=\oplus_{i=1}^n (M_i)_p$ and $N_{q}=\oplus_{j=1}^n (N_j)_{q}$. Thus the diagram above commutes if and only if for all $i$ and $j$, the restrictions of the two compositions to $(M_i)_p$ and $(N_j)_{q}$ are the same, since a linear transformation is determined by what happens on basis elements. This is again equivalent to the following diagram commuting for all $i$ and $j$, where $(\phi_i^j)_p$ is the restriction of $\phi_p$ to $M_i$ and $N_j$:
\[
\begin{tikzcd}
(M_i)_p\ar[d, "(\phi_i^j)_p"]\ar[rr,"(M_i)_{p\to q}"] & &  (M_i)_q\ar[d, "(\phi_i^j)_q"] \\
(N_j)_{p}\ar[rr,swap, "(N_j)_{p \to q}"] & & (N_j)_q
\end{tikzcd}
\]
But the collection of $\phi_p$ forms a morphism if and only if the first diagram commutes for all $p\leq q$, and the restriction of $\phi_p$ to $M_i$ and $N_j$ forms a morphism if and only if the second diagram commutes for all $p\leq q$. Thus we have proved the desired equivalence.
\end{proof}
Observe that the matrix described in \cref{lem:transformation_lemma} is simply $\phi_p\colon \oplus_{i=1}^n (M_i)_p \to \oplus_{j=1}^n (N_j)_{p}$ written as a matrix in the natural way for any $p$ contained in the support of $M_i$ for all $i$. 
\begin{lemma}
\label{lem:int_matrix}
Let $M$, $N$ be direct sums of staircase modules as above and
$\phi\colon M \to N^\eps$
and $\psi\colon N \to M^\eps$ be morphisms.
Then $\phi$ and $\psi$ form an $\eps$-interleaving
if and only if their associated $(n\times n)$-matrices
are inverse to each other.
\end{lemma}
\begin{proof}
The composition $\psi^\eps\circ\phi\colon  M \to M^{2\eps}$ is represented by the matrix
$BA$, as one can see by restricting to a single point contained in all relevant staircases as in the observation above. The morphism $\Sh_M(2\eps)\colon M \to M^{2\eps}$ is represented by the identity matrix. By definition, $(\phi,\psi)$ is an interleaving if and only if these are equal and the corresponding statement holds for $\phi^\eps\circ\psi$, so the statement follows.
\end{proof}

As a consequence, we obtain the following intermediate result

\begin{theorem}
\label{1,3-simplified}
Let $(n,P,Q)$ be a $CI$-instance
and let
$S_1,\ldots,S_n$, $T_1,\ldots,T_n$ be staircases such that
\[
d_S(S_i,T_j)=\begin{cases}
3 & \text{if }(i,j)\in P\\
1 & \text{if }(i,j)\notin P
\end{cases}
\quad\quad
d_S(T_j,S_i)=\begin{cases}
3 & \text{if }(j,i)\in Q\\
1 & \text{if }(j,i)\notin Q
\end{cases}
\]

Write $M_i$, $N_j$ for the modules associated to $S_i$, $T_j$, respectively,
and $M:=\oplus M_i$ and $N:=\oplus N_j$. Then

\[
d_I(M,N)=\begin{cases}
1 & \text{if }(n,P,Q)\in CI\\
3 & \text{if }(n,P,Q)\notin CI
\end{cases}
\]
\end{theorem}
\begin{proof}
Assume first that $(n,P,Q)\in CI$. Let $A$, $B$ be a solution.
We show that $A$ and $B$ define morphisms from $M$ to $N^1$
and from $N$ to $M^1$. 
We restrict to the map from $M$ to $N^1$, as the other case is symmetric.
By Lemma~\ref{lem:matrix_rep_lemma},
it suffices to show that the map from $M_i$ to $N_j^1$
is a morphism. This map is represented by the entry $A_{ij}$.
If $(i,j)\in P$, $A_{ij}=0$ by assumption, and the $0$-map is
always a morphism. If $(i,j)\notin P$,
$d_S(S_i,T_j)=1$ by construction. Hence, by Lemma~\ref{lem:transformation_lemma}
any field element yields a morphism. This shows that $A$ and $B$ define a pair of valid morphisms, and by Lemma~\ref{lem:int_matrix} this pair is an $1$-interleaving, as $AB=I_n$. Also with Lemma~\ref{lem:transformation_lemma},
it can easily be proved that the only morphism $M \to N^\eps$ with $\eps<1$ is the $0$-map. Hence, $d_I(M,N)=1$ in this case.

Now assume that $(n,P,Q)\notin CI$. It is clear that $M$ and $N$ as constructed
are $3$-interleaved: the matrix $I_n$ yields a valid morphism
from $M$ to $N^3$ and from $N$ to $M^3$ with Lemma~\ref{lem:transformation_lemma}.
Assume for a contradiction that there exists an $\eps$-interleaving
between $M$ and $N$ represented by matrices $A$, $B$,
with $\eps<3$. For $(i,j)\in P$, since $d_s(M_i,N_j)=3>\eps$,
Lemma~\ref{lem:transformation_lemma} implies that the entry $A_{i,j}$
must be equal to $0$. Likewise, $B_{j,i}=0$ whenever $(j,i)\in Q$.
By Lemma~\ref{lem:int_matrix}, $AB=I_n$, and it follows that $A$ and $B$ constitute a solution
to the CI-instance $(n,P,Q)$, a contradiction.
\end{proof}

\paragraph{Construction of the staircases.}
To prove Theorem~\ref{1-3}, it suffices
to construct staircases $S_1,\ldots,S_n$,
$T_1,\ldots,T_n$ with the properties
from Theorem~\ref{1,3-simplified},
in polynomial time.

To describe our construction, we consider to two ``base staircases'' which we depict
in Figure~\ref{fig:base_staircase}. 
In what follows, a shift of a point $a$ by $(1,1)$ means replacing $a$ with the point $a-(1,1)$.
The base staircase $S$ is formed
by the points $(-t,t)$ for $t=-4n^2, -4n^2+2, \ldots,4n^2$, but with the right side 
(i.e., the points with negative $t$) shifted by $(1,1)$. Likewise, the base staircase $T$ consists of the same points, but 
with the left side shifted by $(1,1)$. We observe immediately that the staircase distance of the two base
staircases is equal to $1$ in either direction.
We call the points defining the staircases \emph{corners} from now on. 
\begin{figure}\centering

\begin{tikzpicture}[scale=0.5]

\draw[help lines, color=gray!35] (-5,-5) grid (5,5);




\draw[blue,thick] (-8,12) -- (-8,8) -- (-6,8);
\draw[blue, dashed] (-6,8) -- (-6,6) -- (-4,6);
\draw[blue,thick]  (-4,6) -- (-4,4) -- (-2,4) -- (-2,2) -- (0,2) -- (0,0) -- (1,0) -- (1,-1) --  (1,-3) -- (3,-3) -- (3,-5);
\draw[blue, dashed] (3,-5) -- (5,-5) -- (5,-7);
\draw[blue,thick]  (5,-7)  -- (7,-7) --(7,-9) -- (12,-9);


\draw[red,thick] (-9,12) -- (-9,7) -- (-7,7);
\draw[red, dashed] (-7,7) -- (-7,5) -- (-5,5);
\draw[red,thick]  (-5,5) -- (-5,3) -- (-3,3) -- (-3,1) -- (-1,1) -- (0,1) -- (0,0) -- (2,0) --  (2,-2)-- (4,-2) -- (4,-4);
\draw[red, dashed] (4,-4) -- (6,-4) -- (6,-6);
\draw[red,thick]  (6,-6) -- (8,-6) --(8,-8) -- (12,-8);

\node[scale=0.6,left] at (0.1,-0.3) (o) {$(0,0)$};

\node[scale=0.8,right] at (-8,8.4) (S3) {$(-4n^2,4n^2)$};
\node[scale=0.8,left] at (-9,7) (S3) {$(-4n^2-1,4n^2-1)$};
\node[scale=0.8,left] at (7,-9) (S3) {$(4n^2-1,-4n^2-1)$};
\node[scale=0.8,right] at (8,-7.6) (S3) {$(4n^2,-4n^2)$};

\node[scale=0.9, right] at (-8,12) (S) {$S$};
\node[scale=0.9, left] at (-9,12) (S) {$T$};

\end{tikzpicture}
\caption{The base staircases $S$ and $T$.}
\label{fig:base_staircase}

\end{figure}
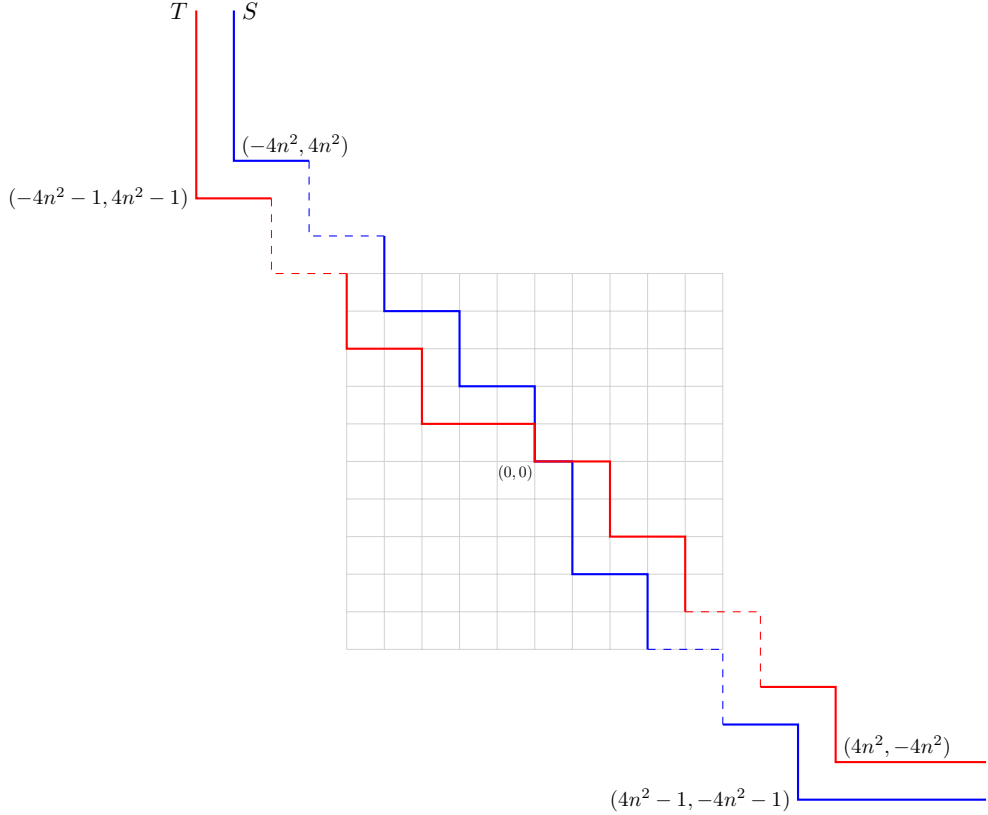

Now we associate to every entry in $P$ a corner in the left side
of $S$ (that is, some $(-t,t)$ with $t>0$). 
We also associate with the entry a corner in $T$, namely the 
shifted point $(-t-1,t-1)$. We do this in a way that between
two associated corners of $S$, there is at least one
corner of the staircase that is not associated. Note that this is
always possible because $|P|\leq n^2$ and we have $2n^2$ corners
on the left side.
We associate corners to entries of $Q$ in the symmetric way, using the right side
of the base staircases. 

We construct the staircases $S_i$ and $T_j$ out of the base staircases $S$
and $T$, only shifting associated corners by $(2,2)$ or $(-2,-2)$ according to
$P$ and $Q$. Specifically, for the staircase $S_i$, 
we start with $S$ and for any entry $(i,j)$ in $P$,
we shift the associated corner of $S$ by $(2,2)$.
For every entry $(j,i)$ in $Q$, we shift the associated corner by $(-2,-2)$.
The resulting (partially) shifted version of $S$ defines $S_i$.

$T_j$ is defined symmetrically: for every $(i,j)\in P$, we shift the associated corner
by $(-2,-2)$. For every $(j,i)\in Q$, we shift the associated corner by $(2,2)$.

We next analyze the staircase distance of $S_i$ and $T_j$. We observe that,
because there is an unassociated corner in-between any two associated corners,
the $\pm (2,2)$ shifts of distinct corners do not interfere with each other.
Hence, it suffices to consider the distance of one associated corner of $S_i$
to $T_j$. Fix the corner $c_S$ of $S$
associated to some entry $(k,\ell)\in P$.
Let $c_T$ denote the associated corner of $T$, that is, $c_T=c_S-(1,1)$.
See Figure~\ref{fig:staircase_shifted} (left) for an illustration.
If $k\neq i$ and $\ell\neq j$, neither $c_S$ nor $c_T$ gets shifted, and since $c_T\leq c_S$,
the shift required from $c_T$ to reach $c_S$ is $0$.
If $k=i$ and $\ell\neq j$, then $c_S$ gets shifted by $(2,2)$, and the required shift is $1$
(see second picture of Figure~\ref{fig:staircase_shifted}).
If $k\neq i$ and $\ell= j$, $c_T$ gets shifted by $(-2,-2)$, the required shift is also $1$
(see $3$rd picture of Figure~\ref{fig:staircase_shifted}).
If $k=i$ and $\ell=j$, both $c_S$ and $c_T$ get shifted, and the distance of the shifted $c_T$
to reach $c_S$ increases to $3$
(see $4$th picture of Figure~\ref{fig:staircase_shifted}).
This argument implies that the (directed) staircase distance
from $S_i$ to $T_j$ is $3$ if $(i,j)\in P$, and $1$ otherwise. A completely symmetric
argument works for $d_S(T_j,S_i)$, inspecting the corners associated to $Q$.

\begin{figure}
\centering
\includegraphics[width=15cm]{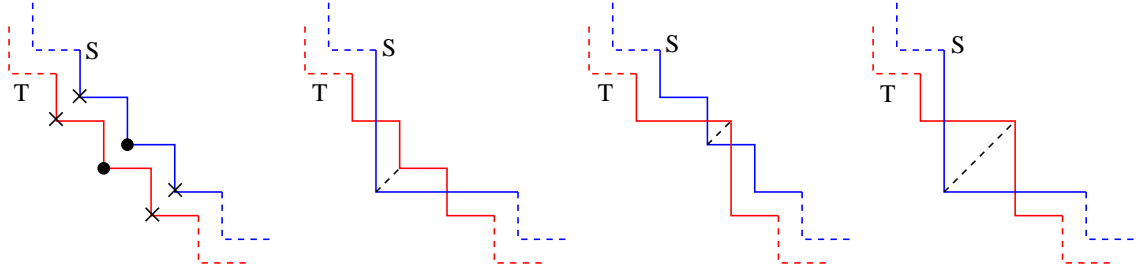}
\caption{Left: The associated corners $c_S$ (on staircase $S$) and $c_T$ (on staircase $T$)
are marked by black circles. The two neighboring corners on both staircases (marked with x)
are not associated and hence not shifted in the construction. 
Second and third picture: the cases $(i,\ell)$ with $\ell\neq j$ and $(k,j)$ with $k\neq i$.
In both cases, the directed staircase distance is $1$, as illustrated
by the dashed line.
Right: The case $(i,j)$. In that case, a shift of $3$ is necessary to move the corner of $T$ to $S$. }
\label{fig:staircase_shifted}
\end{figure}

Finally, it is clear that the size and construction time of each $S_i$ and each $T_j$ is polynomial in $n$.
As remarked at the end of Section~\ref{sec:background}, the staircase module
can be brought in graded matrix representation in polynomial time in $n$,
and the same holds for the direct sum of these modules. 
This finishes the proof of Theorem~\ref{1-3}.

\medskip

With the construction of $M$ and $N$ from Theorem~\ref{1-3} fresh in mind, we can explain the obstacles to obtaining a constant bigger than $3$. Exchanging $3$ with another constant in Theorem~\ref{1,3-simplified} is not a problem; the proof would be exactly the same. The trouble is to construct $S_i$ and $T_j$ satisfying the conditions in Theorem~\ref{1,3-simplified} if $3$ is replaced by some $\eps>3$. In that case, one would have to force $d_S(S_i,T_j)\geq \eps$ for $(i,j)\in P$ and $d_S(T_j,S_i)\geq \eps$ for $(j,i)\in Q$, while still keeping $d_S(S_i,T_j)\leq 1$ for $(i,j)\notin P$ and $d_S(T_j,S_i)\leq 1$ for $(j,i)\notin Q$. 
As we have shown, letting $d_S(S_i,T_j) = 3$ when $(i,j) \in P$ can be done. However, even if $(i,j) \in P$, there might be $i',j'$ such that $(i,j'),(i',j)\notin P$ and $(j',i')\notin Q$, implying
\begin{align*}
d_S(S_i,T_{j'})&\leq 1,\\
d_S(T_{j'},S_{i'})&\leq 1,\\
d_S(S_{i'},T_j)&\leq 1.
\end{align*}
which gives $d_S(S_i,T_j)\leq 3 < \eps$ by the triangle inequality. This proves that one cannot simply increase the constant in Theorem~\ref{1,3-simplified}, change the construction of $S_i$ and $T_j$, and get a better result. That is not to say that using CI problems to improve Theorem~\ref{thm:c_approx_interleaving_is_hard} is necessarily hopeless, but it would not come as a surprise if a radically new approach is needed, if the theorem can be improved at all.

This is related to questions of stability, more precisely of whether $d_B(B(M),B(N)) \leq 3d_I(M,N)$ is true for staircase decomposable modules, where $d_B$ is the bottleneck distance. We have associated pairs of modules to CI problems in a way such that interleavings correspond to solutions of the CI problems. Matchings between the barcodes of the modules (which is what gives rise to the bottleneck distance) correspond to solutions to the CI problems of a particular simple form, namely with a single non-zero entry in each column and row of each matrix. Claiming that $d_B(B(M),B(N)) \leq 3d_I(M,N)$ is then related to claiming that if a CI problem has a solution, then a ``weakening'' of the CI problem has a solution of this simple form. We will not go into details about this, other than to say that there are questions that can be formulated purely in terms of CI problems whose answers could have very interesting consequences for the study of interleavings, also beyond the work done in this paper.



\section{Indecomposable modules}
\label{sec:indecomposable_hardness}

Fix a CI-problem $(n,P,Q)$ as in the previous section and let $M$ and $N$ be the associated persistence modules. We shall now construct two indecomposable persistence modules $\widehat{M}$ and $\widehat{N}$ such that $\widehat{M}$ and $\widehat{N}$ are $\epsilon$-interleaved if and only if $M$ and $N$ are $\epsilon$-interleaved. In what follows we construct $\widehat{M}$; the construction of $\widehat{N}$ is completely analogous.

Recall that a staircase module can be described by a set of generators, or corners. Let $u=(x,x)$ be a point larger than all the corners defining the staircases making up $M$ and $N$. Observe the following: $\dim M_u = n$,  $M_{u\to p}$ is the identity morphism for any $p\geq u$.

Let $s_i = x+7+i/(n+1)$,\footnote{$7$ can be replaced with any $\delta>6$.} for $0 \leq i \leq n+1$. Define $\widehat{M}$ at $p\in \R^2$ as follows
\[\widehat{M}_p = \begin{cases} 
0 & \text{ if } p\geq (s_i, s_{n+1-i}) \text{ for some } 0\leq i\leq n+1,\\
\F & \text { if } p \in [s_i, s_{i+1})\times [s_{n-i}, s_{n-i+1}) \text{ for some } 0\leq i\leq n,\\
M_p & \text{ otherwise. }
\end{cases}
\]
Trivially, $\widehat{M}_{p\to q}$ is the 0 morphism if $\widehat{M}_p =0$ or $\widehat{M}_q = 0$. For $p\leq q$ such that $M_p = \widehat{M}_p$ and $M_q = \widehat{M}_q$, let $\widehat{M}_{p\to q} = M_{p\to q}$, and for $p,q\in [s_i, s_{i+1})\times [s_{n-i}, s_{n-i+1})$ let $\widehat{M}_{p\to q} = 1_\F$. It remains to consider the case that $M_p = \widehat{M}_p$ and $q\in [s_i, s_{i+1})\times [s_{n-i}, s_{n-i+1})$ for some $i$. Observe that all the internal morphisms are fully specified once we define $\widehat{M}_{u\to q}$. Indeed, if $p\geq u$, then $\widehat{M}_{u\to p}$ is the identity, which forces $\widehat{M}_{p\to q} = \widehat{M}_{u\to q}$. For any other $p$ we can always choose an $r\geq p$ such that $r\geq u$ and $M_r = \widehat{M}_r$. The morphism $\widehat{M}_{p\to q}$ is then given by $\widehat{M}_{p\to q} = \widehat{M}_{r\to q} \circ \widehat{M}_{p\to r}  = \widehat{M}_{u\to q} \circ \widehat{M}_{p\to r}$. We conclude by specifying the following morphism

\[\widehat{M}_{u\to q} = 
\begin{cases}
\pi_i \text{ (projection onto coordinate $i$)}&  \text{ if } 1\leq i\leq n\\
\pi_1 + \pi_2 + \ldots + \pi_n& \text{ if } i=0.
\end{cases}
\]
Observe that we have a morphism $\pi^M\colon M\to \widehat{M}$ given by
\[
\pi^M_p = 
\begin{cases}
{\rm id} & \text {if } M_p = \widehat{M}_p\\
\widehat{M}_{u\to p} & \text{otherwise}
\end{cases}
\]
\begin{lemma}
The persistence module $\widehat{M}$ is indecomposable. 
\end{lemma}
\begin{proof}
We recall the following useful trick: if $M$ is \emph{not} indecomposable, say $M\cong M'\oplus M''$, then the projections $M\to M'$ and $M\to M''$ define morphisms which are not given by multiplication with a scalar. Hence, it suffices to show that any endomorphism $\phi\colon \widehat{M}\to \widehat{M}$ is multiplication by a scalar. Furthermore, observe that any endomorphism $\phi$ of $\widehat{M}$ is completely determined by $\phi_{u}$.  
Let $e_i\in \F^n$ denote the vector $(0,\ldots,0, 1,0, \ldots, 0)$ where the non-zero entry appears at the $i$-th index. For $0\leq i\leq n$, $\phi$ must be such that the following diagram commutes
\[
\begin{tikzcd}
\widehat{M}_{u}=\F^n\ar[r]\ar[d, "\phi_{u}"] & \F = \widehat{M}_{(s_i, s_{n-i})}\ar[d, "\phi_{(s_i, s_{n-i})}= \lambda_i"]\\
\widehat{M}_{u}=\F^n\ar[r]& \F = \widehat{M}_{(s_i, s_{n-i})}\\
\end{tikzcd}
\]
For $1\leq i\leq n$ this yields that 
\[\pi_i(\phi_{u}(e_j)) =
\begin{cases}
\lambda_i\in \F & \text{ if } i=j\\
0 & \text { if } i\neq j.
\end{cases}.
\]
In particular, we see that $\phi_{u}(e_i) = \lambda_ie_i$.  For $i=0$ we get 
\[\lambda_0e_i = \lambda_0\cdot (\pi_1+ \ldots + \pi_n)(e_i) = (\pi_1+\ldots+\pi_n)(\lambda_ie_i) = \lambda_ie_i.\]
We conclude that $\lambda_i = \lambda_0$ and that $\phi_{u} = \lambda_0\cdot {\rm id}$. 
\end{proof}

\begin{figure}\centering

\begin{tikzpicture}[scale=1]

\draw[dashed, pattern=horizontal lines, pattern color=black!15] (-6,3) -- (-6,7) -- (-3, 7) --  (-3,6.5) -- (-2.9, 6.5) -- (-2.9, 6.4) -- (-2.8,6.4) -- (-2.8, 6.3) -- (-2.7, 6.3) -- (-2.7, 6.2) -- (-2.6, 6.2) -- (-2.6, 6.1) -- (-2.5, 6.1) --  (-2.5, 6.0) -- (-2,6) -- (-2,3) --  cycle;

\draw[->] (-6,3) -- (-6,7.1);
\draw[->] (-6,3) -- (-1.9,3);

\draw[line width=1mm] (-3.3, 5.7) rectangle (-2.2, 6.8); 
\draw[line width=1mm] (-1, -1) rectangle (7, 7);

\draw[draw opacity=0, pattern=horizontal lines, pattern color=black!15] (-1,-1) -- (-1,7) -- (0, 7) -- (0, 0) -- (7,0) -- (7,-1) --  cycle;

\fill [black!10] (0,4) rectangle (1,5);
\fill [black!10] (1,3) rectangle (2,4);
\fill [black!10] (2,2) rectangle (3,3);
\fill [black!10] (3,1) rectangle (4,2);
\fill [black!10] (4,0) rectangle (5,1);

\draw[draw opacity=0, pattern=horizontal lines, pattern color=black!15]  (0,5) --  (0,4)  -- (1,4) -- (1,3) -- (2,3) -- (2,2) -- (3,2) -- (3,1) -- (4,1) -- (4,0) -- (5,0) -- (0,0) -- cycle;

\draw  (0,5) -- (0,4)  -- (1,4) -- (1,3) -- (2,3) -- (2,2) -- (3,2) -- (3,1) -- (4,1) -- (4,0) -- (5,0);

\draw[dashed]  (0,7) -- (0,5) --  (1,5) node[right] {$(s_1, s_5)$} -- (1,4) -- (2,4) node[right] {$(s_2, s_4)$} -- (2,3) -- (3,3) node[right] {$(s_3, s_3)$}-- (3,2) -- (4,2) node[right] {$(s_4, s_2)$}-- (4,1) -- (5,1) node[right] {$(s_5, s_1)$}-- (5,0) -- (7,0);

\node[scale=1.2] at (0.5,4.5) (a) {$\F$};
\node[scale=1.2]  at (1.5,3.5) (b) {$\F$};
\node[scale=1.2]  at (2.5,2.5) (c) {$\F$};
\node[scale=1.2]  at (3.5,1.5) (d) {$\F$};
\node[scale=1.2] at (4.5,0.5) (e) {$\F$};
\node[scale=1.6] at (0.35,0.3) (m) {$\F^4$};
\node[below] at (0.35, 0.0) (mm) {$u+(7,7)$};
\node[scale=1.6] at (5,5) (o) {$0$};

\draw[->] (m) -- node[scale=0.8, pos=0.6, sloped, below] {$\pi_1+\pi_2+\pi_3+\pi_4~~~$} (a);
\draw[->] (m) -- node[scale=0.8,right] {$\pi_1$} (b);
\draw[->] (m) --  node[scale=0.8,right] {$\pi_2$} (c);
\draw[->] (m) --  node[scale=0.8,below] {$\pi_3$} (d);
\draw[->] (m) --  node[scale=0.8, below] {$\pi_4$} (e);

\end{tikzpicture}
\caption{Left: $\widehat{M}$ coincides with $M$ on the restriction to the shaded subset of $\R^2$. Right: this shows the modification done to $M$ in order to obtain an indecomposable persistence module $\widehat{M}$ in the case $n=4$.   }
\label{fig:indecomp_illu}
\end{figure}
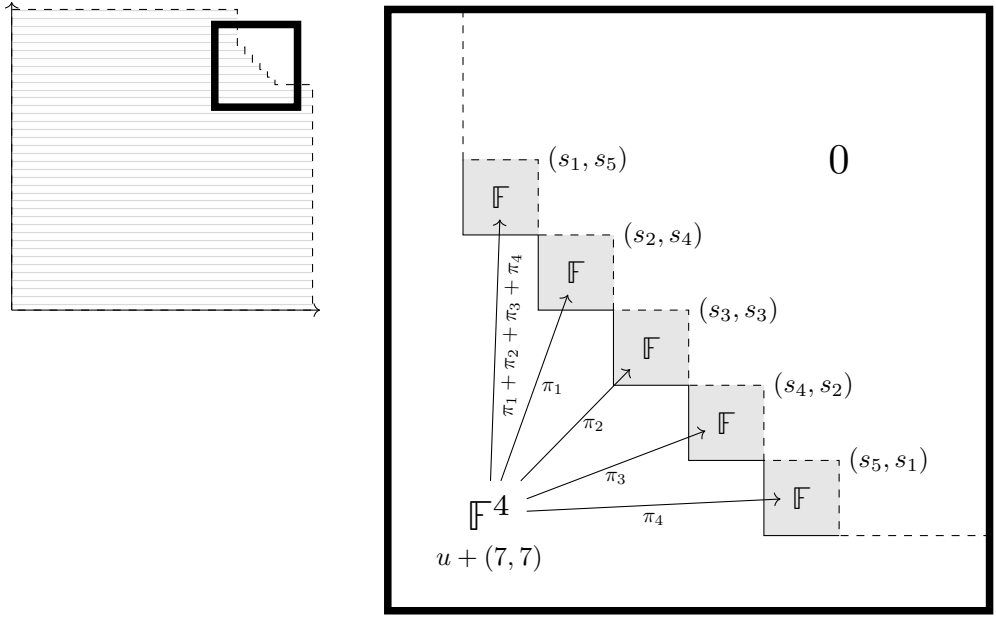

\begin{lemma}
Fix $1\leq \epsilon\leq 3$. $\widehat{M}$ and $\widehat{N}$ are $\epsilon$-interleaved if and only if $M$ and $N$ are $\epsilon$-interleaved.
\end{lemma}
\begin{proof}
Assume that $\phi\colon M\to N^\epsilon$ and $\psi\colon N\to M^\epsilon$ form an $\epsilon$-interleaving pair. Define $\widehat{\phi}\colon \widehat{M}\to \widehat{N}^\eps$ and $\widehat{\psi}\colon \widehat{N}\to \widehat{M}^\epsilon$ by 
\[
\widehat{\phi}_p =
\begin{cases}
\pi^N_{p+\epsilon}\circ \phi_p& \text{ if } M_p = \widehat{M}_p\\
0 &  \text{ otherwise.}
\end{cases},\qquad\qquad \widehat{\psi}_p =
\begin{cases}
\pi^M_{p+\epsilon}\circ \psi_p& \text{ if } N_p = \widehat{N}_p\\
0 &  \text{ otherwise.}
\end{cases}
\]
We will show that these two morphisms constitute an $\epsilon$-interleaving pair. Let $p\in \R^2$ and consider the following two cases:
\begin{enumerate}
\item Assume that $N_{p+\epsilon} = \widehat{N}_{p+\epsilon}$. Under this assumption we have that $\pi^N_{p+\epsilon} = {\rm id}$, and thus $\phi_p = \widehat{\phi}_p$. Using that $\psi$ and $\phi$ form an $\epsilon$-interleaving pair, and that $\pi^M_p = {\rm id}$, we get:
\[\widehat{\psi}_{p+\epsilon}\circ \widehat{\phi}_p = \pi^M_{p+2\epsilon}\circ \psi_{p+\epsilon}\circ \phi_p = \pi^M_{p+2\epsilon}\circ M_{p\to p+2\epsilon} = \widehat{M}_{p\to p+2\epsilon}\circ \pi^M_p = \widehat{M}_{p\to p+2\epsilon}.\]
\item Assume that $N_{p+\epsilon}\neq \widehat{N}_{p+\epsilon}$. Since $\epsilon\geq 1$, it follows by construction that $\widehat{M}_{p+2\epsilon} = 0$. Hence, the interleaving condition is trivially satisfied. 
\end{enumerate}
Symmetrically we get that $\widehat{\phi}_{p+\epsilon}\circ \widehat{\phi}_p = \widehat{N}_{p\to p+2\epsilon}$. Hence, $\widehat{M}$ and $\widehat{N}$ are $\epsilon$-interleaved.

Conversely, assume that $\widehat{\phi}$ and $\widehat{\psi}$ define an interleaving pair between $\widehat{M}$ and $\widehat{N}$. Define $\phi\colon M\to N^\eps$ and $\psi\colon N\to M^\eps$ by 
\[
\phi_p = 
\begin{cases}
\widehat{\phi}_u  & \text{ if } p\geq u\\
\widehat{\phi}_p & \text{ otherwise}
\end{cases},\qquad\qquad 
\psi_p = 
\begin{cases}
\widehat{\psi}_u  & \text{ if } p\geq u\\
\widehat{\psi}_p & \text{ otherwise}
\end{cases}.\qquad\qquad 
\]
By construction,  $\widehat{M}_p = M_p$ and $\widehat{N}_p = N_p$ for all $p<u+(7,7)$. This implies that $\widehat{\phi}_p = \widehat{\phi_u}$ and $\widehat{\psi}_p = \widehat{\psi}_u$ for all $p< u + (7-\epsilon,7-\epsilon)$. Hence, for any $p\leq u$ we must have that \[\psi_{p+\epsilon}\circ \phi_p = \widehat{\psi}_{p+\epsilon}\circ \widehat{\phi}_p = \widehat{M}_{p\to p+2\epsilon} = M_{p\to p+2\epsilon}.\]
Similarly we get that $\phi_{p+\epsilon}\circ \psi_p = N_{p\to p+2\epsilon}$ for all such $p$. In particular, by considering the case $p=u$, we see that $\phi_u$ and $\psi_u$ are mutually inverse matrices. It follows readily that the interleaving condition is satisfied for all $p\not\leq u$. 
\end{proof}

With the two previous results at hand, 
we can state the following corollary 
of \cref{thm:c_approx_interleaving_is_hard}.
\begin{corollary}
\textsc{1-interleaving} is NP-complete and
\textsc{$c$-Approx-Interleaving-Distance} is NP-hard for $c<3$, even 
if the input modules are restricted to indecomposable modules.
\end{corollary}
\begin{proof}
We only prove hardness of \textsc{1-interleaving},
the remaining statements follow with the same methods.
Given a CI-instance $(n,P,Q)$, we use the construction 
from Section~\ref{sec:reduction} to construct two persistence modules $M$ and $N$. Then we transform them into
the indecomposable modules $\widehat{M}$ and $\widehat{N}$ as above.
Note that this transformation can be performed in polynomial time in $n$
by introducing up to $n$ relations at the lower-left corners of the $(n+1)$
rectangles in Figure~\ref{fig:indecomp_illu}.
Hence, an algorithm to decide \textsc{$1$-interleaving} for the case
of indecomposable modules would solve CI in polynomial time.
\end{proof}

\section{One-sided stability}
\label{sec:one_sided}
The results of the previous sections also apply in the setting of one-sided stability. Here we give a brief introduction to the topic; see \cite{bauer2014induced} for a thorough introduction. 

Let $f\colon M\to N$ be a morphism. The linear map $M_{p\to q}$ 
induces a linear map $\ker(f_p)\to\ker(f_q)$ by restriction,
and $N_{p\to q}$ induces a linear map $\coker(f_p)\to\coker(f_q)$
by taking a quotient, as one can readily verify.
We say  that $f$ has \emph{$\epsilon$-trivial kernel} if the map 
$\ker(f_p)\to \ker(f_{p+\epsilon})$ is the $0$-map for all $p\in\R^2$.
Likewise, we say that $f$ has \emph{$\epsilon$-trivial cokernel}
if 
$\coker(f_p)\to \coker(f_{p+\epsilon})$ is the $0$-map for all $p\in\R^2$. If $f$ has $0$-trivial kernel (cokernel), then we say that $f$ is \emph{injective} (\emph{surjective}). The following lemma follows readily from the definition of an $\epsilon$-interleaving.
\begin{lemma}
\label{lem:interleaving_implies_trivialcokernel}
If $f\colon M\to N^\epsilon$ is an $\epsilon$-interleaving morphism (i.e., it forms an $\epsilon$-interleaving with some $g:N\to M^\epsilon$), 
then $f$ has $2\epsilon$-trivial kernel and cokernel. 
\end{lemma}
In fact, Bauer and Lesnick~\cite{bauer2014induced} 
show that in the case of persistence modules over $\R$, 
$M$ and $N$ are $\epsilon$-interleaved if and only if 
there exists a morphism $f\colon M\to N^\eps$ with $2\epsilon$-trivial 
kernel and cokernel. 
They also observe that this equivalence does not generalize to two parameters.
However, it \emph{is} true 
(and the proof is very similar to the one given below) 
that if there exists a morphism $f\colon M\to N^\eps$ with 
$\epsilon$-trivial kernel and cokernel, 
then $M$ and $N$ are $\epsilon$-interleaved. 
Hence, there is a close connection between interleavings and morphisms 
with kernels and cokernels of bounded size 
also in the multi-parameter landscape. 

\begin{lemma}
For any injective $f\colon M\to N^\eps$ with $2\epsilon$-trivial cokernel, there exists a morphism $g\colon N\to M^\eps$ such that $f$ and $g$ constitute an $\epsilon$-interleaving pair.
\label{lem:mono-1}
\end{lemma}
\begin{proof}
We have the following commutative square for all $p\in \R^2$:
\[
\begin{tikzcd}
M_p\ar[d, "f_p"]\ar[rr,"M_{p\to p+2\epsilon}"] & &  M_{p+2\epsilon}\ar[d, "f_{p+2\epsilon}"] \\
N_{p+\epsilon}\ar[rr,swap, "N_{p+\epsilon\to p+3\epsilon}"] & & N_{p+3\epsilon}
\end{tikzcd}
\]
Let $n\in N_{p+\eps}$. Since $f$ has $2\epsilon$-trivial cokernel and $f$ is injective, there exists a unique $m\in M_{p+2\epsilon}$ such that $f_{p+2\epsilon}(m) = N_{p+\epsilon\to p+3\epsilon}(n)$. Define $g_p\colon N_{p+\epsilon} \to M_{p+2\epsilon}$ by $g_p(n) = m$. Doing this for all $p\in \R^2$ defines a morphism $g\colon N^\eps\to M^{2\epsilon}$ and we leave it to the reader to verify that $f$ and $g^{-\epsilon}$ define an $\epsilon$-interleaving pair. 
\end{proof}
\smallskip

For fixed parameters $s,t\in[0,\infty]$, 
we consider the following computational problem:

\begin{quote}
\textsc{$s$-$t$-trivial-morphism}: Given two persistence modules $M$, $N$ in 
graded matrix representation, decide whether there exists a morphism $f\colon M\to N$
with $s$-trivial kernel and $t$-trivial cokernel.
\end{quote}

Choosing $s=t=0$ simply asks whether the modules are isomorphic,
which can be decided in polynomial time~\cite{brooksbank2008testing}.
On the other extreme, $s=t=\infty$ imposes no conditions
on the morphism, which turns the decision problem to be trivially true,
using the $0$-morphism. We show

\begin{theorem}
\label{thm:one_sided_hardness}
\textsc{$s$-$t$-trivial-morphism} is NP-complete for every $(s,t)\notin\{(0,0),(\infty,\infty)\}$.
\end{theorem}

The case $(s,t)$ is computationally equivalent to the case $(cs,ct)$
with $c>0$, since we can scale all grades occurring in $M$ and $N$
by a factor of $c$. So, it suffices to prove hardness of
\textsc{$2$-$t$-trivial morphism}, \textsc{$s$-$2$-trivial morphism}
(we will see that the choice of $2$ will be convenient in the argument), \textsc{$\infty$-$0$-trivial morphism} and  \textsc{$0$-$\infty$-trivial morphism}. 

Note that for any choice of $s$ and $t$, \textsc{$s$-$t$-trivial-morphism}
is in NP. The argument is similar to the first part 
of the proof of~\cref{thm:1_interleaving_np_complete}: a morphism
can be specified in polynomial size with respect to the module sizes,
and we can check the triviality conditions of the kernel and cokernel
by considering ranks of matrices.

For the hardness, 
we first focus on the case $(s,2)$, hence we want to decide the existence
of a morphism with $s$-trivial kernel and $2$-trivial cokernel.
The following simple observation is the key insight of the proof.

\begin{lemma}
\label{mono_lemma_for_staircase}
Let $M$, $N$ be as in \cref{1-3}. Any morphism $f\colon M\to N^1$ with $2$-trivial cokernel
is injective.
\end{lemma}
\begin{proof}
Recall that both $M$ and $N$ are direct sums of $n$ staircase modules. Let $p$ be any point such that $\dim M_p = \dim N_p = n$, and observe that $M_{p\to q} = {\rm id}_\F$, $N_{p\to q} = {\rm id}_\F$ and $f_p = f_q$ for all $q\geq p$. 
In particular, if $q=p+(2,2)$, the induced map $\coker(f_p)\to\coker(f_q)$
is the identity, and since $f$ has a $2$-trivial cokernel by assumption,
the map is also the $0$-map. Hence $\coker(f_p)$ is trivial,
implying that the map $f_p$ is surjective, and hence also injective,
and the same holds for $f_q$ with $q\geq p$.

Now consider $f_r$ for an arbitrary $r\in\R^2$.
Let $q\geq r$ be a point satisfying $q\geq p$.
Since the internal morphisms of $M$ are all injective and $f_p$ is injective,
so is $f_r$.
\end{proof}

In other words, for $M$ and $N^1$ as above, 
the answer to \textsc{$s$-$2$-trivial-morphism} is independent of $s$.
Moreover, it follows:

\begin{corollary}
\label{cor:reduction_trivial_interleaving}
With $M,N$ as above, there exists a morphism $f\colon M\to N^1$
with $2$-trivial cokernel and $s$-trivial kernel if and only if
$M$ and $N$ are $1$-interleaved.
\end{corollary}
\begin{proof}
If such a morphism exists, Lemma~\ref{mono_lemma_for_staircase} 
guarantees that the morphism
is in fact injective with $2$-trivial cokernel. 
Lemma~\ref{lem:mono-1} with $\eps=1$ guarantees that the modules are 
$1$-interleaved.

Vice versa, if $M$ and $N$ are $1$-interleaved, there is 
a morphism $f$ with $2$-trivial kernel and cokernel
by Lemma~\ref{lem:interleaving_implies_trivialcokernel}. Again using 
Lemma~\ref{mono_lemma_for_staircase} guarantees that $f$ is injective,
hence has a $0$-trivial kernel.
\end{proof}

\begin{corollary}
\label{cor:s-2-hardness}
\textsc{$s$-$2$-trivial-morphism} is NP-hard for all $s\in[0,\infty]$.
\end{corollary}
\begin{proof}
Given a CI-instance, we transform it into modules $M$ and $N$ as in \cref{sec:reduction}.
Assuming a polynomial time algorithm for \textsc{$s$-$2$-trivial-morphism},
we apply it on $(M,N^1)$. If the algorithm returns that a morphism
exists, we know by~\cref{cor:reduction_trivial_interleaving}
that $M$ and $N$ are $1$-interleaved and therefore, the CI-instance
has a solution. If no morphism exists, $M$ and $N$ are not $1$-interleaved
and therefore, the CI-instance has no solution. We can thus solve the
CI problem in polynomial time.
\end{proof}

\paragraph{Dual staircases.} We will prove that \textsc{$2$-$t$-trivial-morphism} is NP-hard by a reduction from \textsc{$s$-$2$-trivial-morphism}. First we need some notation. For a staircase $S$, let $S^\circ$ denote the interior of $S$, and for a staircase module $M_l$ supported on a staircase $S$, we let $M_l^\circ$ denote the interval module supported on $S^\circ$. Observe that there is a canonical injection $M_l^\circ \hookrightarrow M$ (given by $m\mapsto m$).  It is also easy to see that $d_s(S,T) = d_s(S^\circ, T^\circ)$. Here $d_s$ for interiors of staircases is defined in the obvious way. The reason why we look at interiors is technical: We eventually end up with a dual module $(M^\circ)^*$, and taking interiors makes sure the changes in this dual module happen at given points instead of ``immeadiately after'' the points, which is needed for a graded matrix representation of the module.

\begin{lemma}
Let $M$ and $N$ be staircase decomposable modules. There exists an injection $f\colon M\to N$ with $\epsilon$-trivial cokernel if and only if there exists an injection $f^\circ\colon M^\circ \to N^\circ$ with $\epsilon$-trivial cokernel. 
\end{lemma}
\begin{proof}
Let $M=\oplus_i M_i$ and $N=\oplus_j N_j$. Observe that $S\subseteq T$ if and only if $S^\circ \subseteq T^\circ$. Therefore, any morphism $M^\circ_i\to N^\circ_j$ extends to a morphism $M_j\to N_j$ in the obvious way. Conversely, any morphism $M\to N$ restricts to a morphism $M^\circ \to N^\circ$. It is not hard to see that extension and restriction are inverse functions. In particular, there is a one-to-one correspondence between morphisms $f\colon M\to N$ and $f^\circ\colon M^\circ \to N^\circ$. 

Suppose $f^\circ$ is injective. For any point $p$, there exists a $\delta>0$ such that 
$M_{p\to p+\delta}$ and $N_{p\to p+\delta}$ are isomorphisms, which also gives $f_{p+\delta}^\circ = f_{p+\delta}$. Since $f^\circ$ is injective, $f_{p+\delta}^\circ= f_{p+\delta}$ is, and by using the isomorphisms, we get that $f_p$ is injective, too. Since $p$ was arbitrary, we conclude that $f$ is injective. The converse can be proved by using the dual fact that for any $p$, there exists a $\gamma$ such that $M_{p-\gamma\to p}^\circ$ and $N_{p-\gamma\to p}^\circ$ are isomorphisms.

Suppose that $\coker f$ is not $\eps$-trivial, so there is a $p$ and an $m\in N_p$ such that $N_{p\to p+\eps}(m)$ is not in the image of $f_{p+\eps}$. Similarly to how we picked $\delta$ above, we can pick $\delta$ and $\gamma$ with $\delta\leq \gamma$ in a way that makes the following diagram commute, with equalities and isomorphisms as shown. 
\[
\begin{tikzcd}
&M_{p+\delta}^\circ\ar[d, "="]\ar[rrr] &&& M_{p+\eps+\gamma}^\circ\ar[d, "="]\\
M_p\ar[d, "f_p"]\ar[r,"\cong"] &  M_{p+\delta}\ar[d, "f_{p+\delta}=f_{p+\delta}^\circ"]\ar[rr]& & M_{p+\eps}\ar[d, "f_{p+\eps}"]\ar[r,"\cong"] & M_{p+\eps+\gamma}\ar[d, "f_{p+\eps+\gamma}=f_{p+\eps+\gamma}^\circ"] \\
N_{p}\ar[r,"\cong"] & N_{p+\delta}\ar[rr]\ar[d, "="]& & N_{p+\eps}\ar[r,"\cong"] & N_{p+\eps+\gamma}\ar[d, "="]\\
&N_{p+\delta}^\circ\ar[rrr] &&& N_{p+\eps+\gamma}^\circ\\
\end{tikzcd}
\]
All the horizontal maps are internal morphisms. We know that $N_{p\to p+\eps}(m)\in N_{p+\eps}$ is not in the image of $f_{p+\eps}$. Let $m'\in N_{p+\eps+\gamma}^\circ$ be the image of $m$ along the maps in the above diagram.  Then $m'$ is in the image of $N_{p+\delta\to p+\eps+\gamma}^\circ$, but not in the image of $f_{p+\eps+\gamma}^\circ$. Since $(\eps+\gamma)-\delta\geq \eps$, this shows that $f^\circ$ is not $\eps$-trivial. Again, the argument can be dualized to show the converse.
\end{proof}

For an interval $I\subseteq \R^2$, define the dual interval $I^\ast$ as follows: $(x,y)\in I^\ast$ if and only if $(-x, -y)\in I$. And for an interval module $M_l$ supported on $I$, let $M_l^\ast$ denote be the interval module supported on $I^\ast$. If $M=\oplus_i M_i$ is a sum of interval modules $M_i$, then $M^\ast = \oplus_i (M_i)^\ast$. This is equivalent to considering $M$ as a module indexed by $\R^2$ with the partial order reversed. 

Let $M^\circ =\oplus_i M^\circ_i$ and $N^\circ =\oplus_j N^\circ_j$, where $M_i^\circ$ and $N_j^\circ$ are interval modules supported on interiors of staircases, and let $f^\circ\colon M^\circ\to N^\circ$. Observe that we can represent $f^\circ$ by a collection of matrices $\{A_p\}_{p\in \R^2}$, where $A_p$ is the matrix representation of $f^\circ_p$ with respect to the bases given by the non-trivial elements of $\{(M_i^\circ)_p\}_i$ and $\{(N_j^\circ)_p\}_j$. Similarly, for any $p\leq q$, we can represent the linear maps $M_{p\to q}^\circ$ and $N_{p\to q}^\circ$ by matrices with respect to the obvious bases.

Importantly, representing $f^\circ_p$ by matrices $A_p$ as above, we get a dual morphism $(f^\circ)^\ast\colon N^\ast \to M^\ast$ given by the matrices $\{(A_{-p})^T\}_{p\in \R^2}$. This induces a bijection between the set of morphisms from $M^\circ$ to $N^\circ$ and the set of morphisms from $(N^\circ)^\ast$ to $(M^\circ)^\ast$.
\begin{lemma}
$f^\circ$ is an injection with $\epsilon$-trivial cokernel if and only if $(f^\circ)^\ast$ is a surjection with $\epsilon$-trivial kernel. 
\end{lemma}
\begin{proof}
The first part is straightforward: the matrix $A_p$ represents a surjective linear map if and only if $A_p^T$ represents an injective linear map. Since $(f^\circ)^\ast_p = f^\circ_{-p}$, the result follows readily.  

For the second part, let $p$ be any point in $\R^2$, and let $X$ be the matrix representation of the morphism $N^\circ_{p\to p+\epsilon}$ with respect to the basis given by the $N^\circ_j$'s. Then, by construction, $X^T$ is a matrix representation for $(N^\circ)^\ast_{-p-\epsilon\to -p}$ (with respect to the dual bases). Using the elementary fact that ${\rm col}(X)\subseteq {\rm col}(A_{p+\epsilon})$ if and only if $\ker(A_{p+\epsilon}^T)\subseteq \ker(X^T)$, where ${\rm col}(X)$ denotes the column space of $X$, we conclude that ${\rm im}(N^\circ_{p\to p+\epsilon})\subseteq {\rm im}(f^\circ_{p+\epsilon})$ if and only if $\ker ((f^\circ)^\ast_{-p-\epsilon})\subseteq \ker ((N^\circ)^\ast_{-p-\epsilon\to p})$. As $p$ was arbitrary, this concludes the proof. 
\end{proof}

\begin{corollary}
\label{cor:2-t-hardness}
\textsc{$2$-$t$-trivial-morphism} is NP-hard for all $t\in[0,\infty]$.
\end{corollary}
\begin{proof}
This follows from the previous two lemmas and \cref{cor:s-2-hardness}. There is however a technical obstacle arising from the fact that $(M^\circ)^\ast$
and $(N^\circ)^\ast$ have their generators at grade 
$(-\infty,-\infty)$. This problem is easy to solve, either by
altering the graded matrix representation to allow such a generator,
or by placing all generators at a sufficiently small value $p\in\R^2$
that is smaller than all corners of the staircase, see \cref{fig:dual-I} for an illustration.
Introducing such a minimal grade does not invalidate any of the
given arguments~-- we omit the technical details.
\end{proof}

\begin{figure}\centering

\begin{tikzpicture}[scale=0.9]

\draw[dashed, fill=black!15] (0,4) -- (5, 4) --  (5,0) --  (2,0) -- (2,1) -- (1,1) -- (1,2) -- (0,2) -- cycle;

\draw (5,0) --  (2,0) -- (2,1) -- (1,1) -- (1,2) -- (0,2) -- (0,4);

\node[left] at (2,0) (g1) {$g_1$};
\node[left] at (1,1) (g2) {$g_2$};
\node[left] at (0,2) (g3) {$g_3$};

\node[right] at (2,1) (r1) {$r_1$};
\node[right] at (1,2) (r2) {$r_2$};
\node[scale=2] at (4,3) (a) {$I$};

\draw[dashed] (-4,4) -- (4,-4);
\node[right] at (-3.9,4) (y) {$y=-x$};

\begin{scope}[xscale=-1, yscale=-1]
\draw[dashed, fill=black!15] (0,4) -- (5, 4) --  (5,0) --  (2,0) -- (2,1) -- (1,1) -- (1,2) -- (0,2) -- cycle;

\node[right] at (2,0) (g1) {$r^*_1=-g_1$};
\node[right] at (1,1) (g2) {$r^*_2=-g_2$};
\node[right] at (0,2) (g3) {$r^*_3=-g_3$};

\node[left] at (2,1) (r1) {$r_4^*=-r_1$};
\node[left] at (1,2) (r2) {$r_5^* = -r_2$};
\node[left] at (5,4) (g0) {$g_1^* = (-\infty, -\infty)$};
\node[left] at (5,0) (g4) {$r_6^* = (-\infty, (-g_1)_2)$};
\node[right] at (0,4) (g5) {$r_7^* = ((-g_3)_1, -\infty)$};

\node[scale=2] at (4,3) (a) {$(I^\circ)^*$};
\end{scope}
\end{tikzpicture}
\caption{The staircase module $k_I$ supported on the interval $I$  admits a graded matrix representation with $G=\{g_1,g_2,g_3\}$ and $R=\{r_1, r_2\}$. 
The module $(k_{I^\circ})^* = k_{(I^\circ)^*} $ admits a (generalized) graded matrix representation with $G^*=\{(-\infty, -\infty)\}$ and $R^*=\{-g_1, -g_2, -g_3, -r_1, -r_2, (-\infty, (-g_1)_2), ((-g_3)_1, -\infty)\}$. In the proof of \cref{cor:2-t-hardness} we may replace $\infty$ with $z\gg 0$ to obtain a proper graded matrix representation. }
\label{fig:dual-I}
\end{figure}
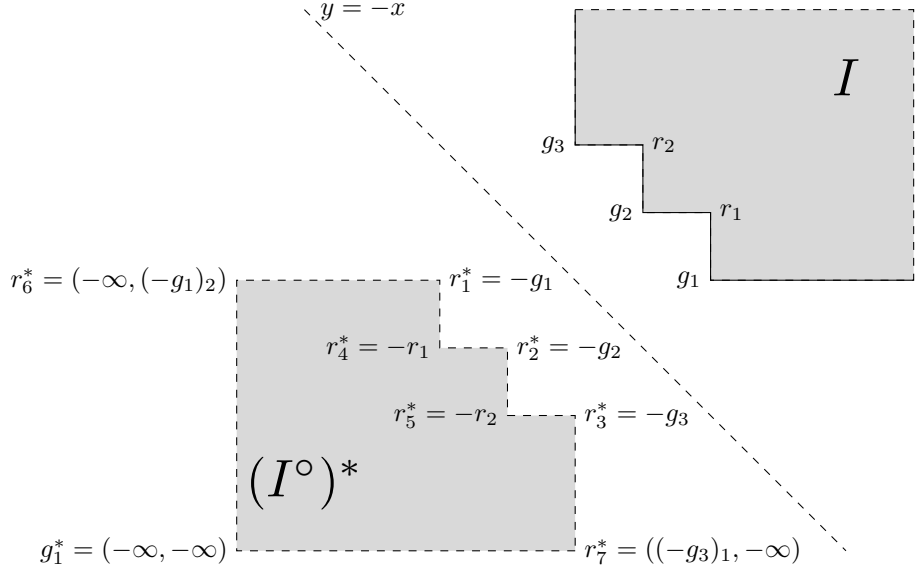

\paragraph{Surjective morphisms.}
After \cref{cor:s-2-hardness} and~\cref{cor:2-t-hardness}, all we have left to prove \cref{thm:one_sided_hardness} is the cases \textsc{$\infty$-$0$-trivial-morphism} and \textsc{$0$-$\infty$-trivial-morphism}. Recall that these correspond to asking for a surjection in the first case and an injection in the second.

\begin{lemma}
\label{lem:inf-0-hardness}
\textsc{$\infty$-$0$-trivial-morphism} and \textsc{$0$-$\infty$-trivial-morphism} are both NP-hard.
\end{lemma}

\begin{proof}
We will only prove the first case; the second follows by dualizing the arguments in an appropriate way, for instance by using dual staircases as above.

Recall that we have assumed $\F$ to be finite. Let $q$ denote the number of elements in $\F$, and assume that $\phi$ is a 3CNF formula with $n$ variables $\{x_1,\dots,x_n\}$ and $m$ clauses $\{c_1,\dots,c_m\}$. We shall construct modules $M= A \oplus B \oplus (\oplus_{i=1}^{n} \oplus_{r=1}^{q} M_i^r)$ and $N=N_1 \oplus N_2$, where $M_i^r$, $A$, $B$, $N_1$ and $N_2$ are staircase modules, in such a way that there exists a surjection $M \to N$ if and only if $\phi$ is satisfiable. Importantly, we know from \cref{lem:matrix_rep_lemma} that any morphism between staircase decomposable modules can be represented by a matrix with entries in $\F$. We only stated the result in the case where each module is built from the same number of staircases, but the same argument shows that a morphism $M\to N$ in this case is described by a $2 \times (nq+2)$-matrix, which we shall assume   is ordered in the following way

\[
\begin{blockarray}{ccccccccc}
A & B & M_{1}^1 & M_1^2 & \cdots & M_2^1 & \cdots & M_{n}^q&   \\
\begin{block}{(cccccccc)c}
  \free & \free & \free & \free & \cdots &\free  & \cdots&  \free& N_1 \\
  \free & \free & \free & \free & \cdots   & \free  & \cdots&  \free&N_2 \\
\end{block}
\end{blockarray}
 \]
Furthermore, 
recall that any staircase module is defined by a set of generators, i.e. a set of incomparable points defining the ``corners'' of the staircase. It is not hard to see that a morphism $M\to N$ is surjective if and only if it is surjective at the all the corners points of $N_1$ and $N_2$. 

Let \[D=\{A,B, N_1, N_2\}\bigcup_{i,r}  \{M_i^r\}\] and let $S\subseteq \R^2$ be a set of pairwise incomparable points. Any function $G\colon S\to P(D)$, where $P(D)$ is the power set of $D$, specifies the modules in the decomposition of $M$ and $N$ by enforcing that $X\in D$ has a corner point at $s\in S$ if and only if $X\in G(s)$. In what follows we shall define such a function $G$ in four steps, and define the staircase modules in $D$ accordingly. 

Let $S =\{a,b,g_i^r,g_i^{r,s},h_j^{y,z,w}\}$ be a set of distinct incomparable points in $\R^2$, where $i,j,r,s,y,z,w$ run through indices which will be defined as we define $G$. In the initial step we define $G(a)=\{A, N_1\}$ and $G(b)=\{B, N_2\}$. The addition of these corners enforce that the matrix (in the ordering given above) must be of the form
\[
\left(
\begin{array}{ccccc}
1&0&*&\dots&*\\
0&1&*&\dots&*\\
\end{array}
\right).
\]
This can be seen as follows: since $a$ and $b$ are incomparable, and $a\in A$ while $a\notin N_2$, we must have that $N_2 \nsubseteq A$. \cref{lem:transformation_lemma} allows us to conclude that the only morphism from $A\to N_2$ is the trivial one.  Similarly we see that the morphism $B\to N_1$ must be the trivial one. Furthermore, since $M_a = A_a$ and $N_a = (N_2)_a$, surjectivity at $a$ implies that $A\to N_1$ must be nonzero, which gives the nonzero entry in the first column. We can multiply any column in the matrix with a nonzero element without changing the validity or surjectivity of the morphism, so we can assume that this element is $1$.  Similarly we get a 1 in the second row of the second column.

We proceed our inductive step by defining $G(g_i^r) = \{A,M_i^r,N_1,N_2\}$ for all  $1\leq i\leq n$ and $1 \leq r\leq q$. Restricting the matrix to the columns corresponding to $A$ and $M_i^r$ we get 
\[
\left(
\begin{array}{cc}
1&*\\
0&*\\
\end{array}
\right).
\]
For the morphism to be surjective at the point $g_i^r$, this matrix must be of full rank. Therefore we can write it as
\[
\left(
\begin{array}{cc}
1&d_i^r\\
0&1\\
\end{array}
\right),
\]
where we again have used the fact that we can scale columns by nonzero constants. In other words, any surjection $M\to N$ must be of the form
\[
\left(
\begin{array}{cccccccccccc}
1&0&d_1^1&\dots&d_1^q&d_2^1&\dots&d_2^q&\dots&d_n^1&\dots&d_n^q\\
0&1&1&\dots&1&1&\dots&1&\dots&1&\dots&1\\
\end{array}
\right).
\]
Continuing, let $G(g_i^{r,s}) = \{M_i^r,M_i^s,N_1,N_2\}$, for all $1\leq i\leq n$ and $1 \leq r<s\leq q$.
Restricting the matrix to the columns corresponding to $M_i^r$ and $M_i^s$ yields the matrix
\[
\left(
\begin{array}{cc}
d_i^r&d_i^s\\
1&1\\
\end{array}
\right).
\]
For the matrix to be surjective at $g_i^{r,s}$, also this matrix must be of full rank. In particular, it must be the case that $d_i^r \neq d_i^s$, and therefore exactly one of $d_i^1, \dots, d_i^q$ equals 0.  We will interpret $d_i^1=0$ as choosing $x_i$ to be false, and $d_i^1\neq 0$ as choosing $x_i$ to be true. 

What remains is to encode the clauses of $\phi$. For a clause $c_j$, let $x_{\alpha_{j,1}}, x_{\alpha_{j,2}}, x_{\alpha_{j,3}}$ be the variables such that either the variable itself or its negation occurs in $c_j$, with $\alpha_{j,1}<\alpha_{j,2}<\alpha_{j,3}$. For $1\leq i\leq3$, let $X_j^i =\{1\}$ if $x_{\alpha_{j,i}}$ occurs in $c_j$; if instead its negation occurs, let $X_j^i =\{2,\dots,q\}$. For example, if $c_j = x_1 \vee \neg x_2 \vee \neg x_4$, then $\alpha_{j,1}=1$, $\alpha_{j,2}=2$ and $\alpha_{j,3}=4$, and $X_j^1 =\{1\}$, $X_j^2 =\{2,\dots,q\}$ and $X_j^3 =\{2,\dots,q\}$. Define $G(h_j^{y,z,w}) = \{B,M_{\alpha_{j,1}}^r,M_{\alpha_{j,2}}^s,M_{\alpha_{j,3}}^t,N_1,N_2\}$, for all $1\leq j\leq m$ and $y\in X_j^1$, $z\in X_j^2$, $w\in X_j^3$.

This time, the following submatrix must have rank $2$ for all $h_j^{y,z,w}$ with $j,y,z,w$ as above.
\begin{equation}
\label{mat:onesided}
\left(
\begin{array}{cccc}
0&d_{\alpha_{j,1}}^y&d_{\alpha_{j,2}}^z&d_{\alpha_{j,3}}^w\\
1&1&1&1\\
\end{array}
\right)
\end{equation}
At this stage we have concluded the construction of the modules and no further restrictions will be imposed on the matrix. In particular, the above shows that there exists a surjection $M\to N$ if and only if there is an assignment $d_i^r\in \F$ such that the following is satisfied: 
\begin{itemize} 
\item $\F= \{d_i^1, \ldots, d_i^q\}$ for all $1\leq i\leq n$.
\item The matrix of \cref{mat:onesided} has full rank for every $h_j^{y,z,w}$.
\end{itemize}
We show that this is equivalent to $\phi$ being satisfiable. 

``$\Rightarrow$'': Assume that $\phi$ is satisfiable and pick a satisfying assignment. If $x_i$ is set to false, then define $d_i^1=0$. If $x_i$ is set to true, then define $d_i^2=0$. In both cases we assign the remaining variables values such that $\F= \{d_i^1, \ldots, d_i^q\}$ for all $1\leq i\leq n$. Consider the clause $c_j$ as above, and assume that $x_{\alpha_{j,l}}$ is assigned a truth value such that the literal associated to $x_{\alpha_{j,l}}$ in $c_j$ evaluates to true. Then $d_{\alpha_{j,l}}^y \neq 0$ for all $y \in X_j^l$, implying that the matrix of \cref{mat:onesided} has rank $2$ for all $h_j^{y,z,w}$. 

``$\Leftarrow$'': Assume an assignment of the variables $d_i^r$ satisfying the two bullet points above, and set $x_i$ to be false if $d_i^1=0$, and true otherwise. Consider the clause $c_j$ as above, and observe that there exists an index $y\in X_j^l$ such that $d_{\alpha_{j,l}}^y=0$ if and only if the literal in $c_j$ associated to $x_{\alpha_{j,l}}$ evaluates to false. In particular, $c_j$ evaluates to true if and only if at least one of $d_{\alpha_{j,1}}^y, d_{\alpha_{j,2}}^z$ and $d_{\alpha_{j,3}}^w$ is non-zero for every $(x,y,z)\in X_j^1\times X_j^2\times X_j^3$.  This is equivalent to the matrix of \cref{mat:onesided} having full rank for every $h_j^{y,z,w}$.

In the end, we have a reduction from 3SAT to \textsc{$\infty$-$0$-trivial-morphism}. To complete the proof, we must show that the instance of \textsc{$\infty$-$0$-trivial-morphism} can be constructed in polynomial time in the input size of the instance of 3SAT. As we have assumed $q$ to be fixed and finite, it suffices to observe that $M$ is defined by $nq+2$ staircase modules, while $N$ is a sum of  $2$ staircase modules, and that each of these are generated by at most $2+nq+n{q\choose2}+m(q-1)^3$ generators. We remark that the generators can be chosen along the antidiagonal $x=-y$ in $\R^2$.
\end{proof}

An interesting point is that the number of generators of the staircases in the proof increases with the size of $\F$. Hence, the proof strategy only applies in the setting of a finite field (with a constant number of elements). 

We conclude this section by remarking that \textsc{$0$-$\infty$-trivial-morphism} is equivalent to the problem of deciding if a module $M'$ is a \emph{submodule} of another persistence module $M$. Interestingly, it can be checked in polynomial time if $M'$ is a \emph{summand} of $M$ \cite[Theorem~3.5]{brooksbank2008testing}.

\section{A distance induced by a noise system}
\label{sec:noise}
As a last application of our methods, we show that a particular distance induced by a \emph{noise system} is NP-hard to approximate within a factor of 2. 

A noise system, as introduced by Scolamiero et al. \cite {scolamiero2017multidimensional}, induces a pseudometric on (tame) persistence modules. In this section we shall briefly consider one particular noise system and we refer the reader to \cite{scolamiero2017multidimensional} for an in-depth treatment of the more general theory.

We say that $f\colon M\to N$ is a \emph{$\mu$-equivalence} if $f$ has $\mu_1$-trivial kernel and $\mu_2$-trivial cokernel, and $\mu_1+\mu_2\leq \mu$. From this definition we can define the following distance between two persistence modules $M$ and $N$
\begin{align*}
\dnoise(M,N) = \inf\{\mu \mid \exists M \xleftarrow{f} X \xrightarrow{g} N, f \text{ an } \eps\text{-equivalence, }g \text{ a } \delta\text{-equivalence and } \eps+\delta\leq\mu \}
\end{align*}
The reader may verify that this distance coincides with the distance induced by the noise system $\{S_\epsilon\}$ where $S_\epsilon$ consists of all persistence modules $M$ with the property that $M_{p\to p+(\epsilon, \epsilon)}$ is trivial for all $p$. In particular, $\dnoise$ is indeed an extended pseudometric  \cite[Proposition 8.7]{scolamiero2017multidimensional}.

Like for the interleaving distance, we can define the computational problem of $c$-approximating $d$ for a constant $c\geq 1$.
\begin{quote}
	\textsc{$c$-Approx-$\dnoise$}: 
	Given two persistence modules $M$, $N$ in 
	graded matrix representation, return a real number $r$ such that
	\[\dnoise(M,N)\leq r \leq c\cdot\dnoise(M,N)\]
\end{quote}
\begin{theorem}
\textsc{$c$-Approx-$d$} is NP-hard for $c<2$.
\end{theorem}
\begin{proof}
Let $(n,P,Q)$ be a CI-instance and construct $M$ and $N$ as in Theorem~\ref{1-3}. We will show the following implications:
\begin{align*}
d_I(M,N) = 1 &\Rightarrow \dnoise(M,N^1) \leq 2\\
d_I(M,N) = 3  &\Rightarrow \dnoise(M,N^1) \geq 4.
\end{align*}
This allows us to conclude that an algorithm $c$-approximating $\dnoise(M,N)$ for $c<2$ will return a number $<4$ if $d_I(M,N)=1$ and a number $\geq 4$ if $d_I(M,N)=3$. This constitutes a polynomial-time reduction from CI to \textsc{$2$-Approx-$\dnoise$} and the result follows from \cref{thm:ci_np_complete}.

First assume that $d_I(M,N) = 1$. Let $X=M$ with $f:X\to M$ the identity morphism.  
\cref{lem:interleaving_implies_trivialcokernel} shows that the interleaving morphism $g\colon M\to N^1$ has $2$-trivial cokernel, and from \cref{mono_lemma_for_staircase} we know that it is injective. Hence $g$ is a $2$-equivalence and thus $\dnoise(M,N^1)\leq 2$.

Now assume that $\dnoise(M,N^1) < 4$. By definition this gives a diagram $M \xleftarrow{f} X \xrightarrow{g} N^1$ where $f$ is an $\eps$-equivalence, $g$ is a $\delta$-equivalence, and $\eps+\delta<4$. We may assume that both $f$ and $g$ are injective. To see this, consider $x\in \ker(f_p)$. Because $f$ is an $\eps$-equivalence, $\ker(f)$ is $\eps$-trivial, so $X_{p\to p+(\eps,\eps)}(x) = 0$. This gives 
$$0=g_{p+(\eps,\eps)}\circ X_{p\to p+(\eps,\eps)}(x) = N^1_{p\to p+(\eps,\eps)} \circ g_p(x).$$ 
Since $N^1_{p\to p+(\eps,\eps)}$ is injective we conclude that $g_p(x) =0$. This shows that  $\ker(f) \subseteq \ker(g)$, and by symmetry, that $\ker(f) = \ker(g)$. Replacing $X$ with $\tilde{X} = X/\ker(f)$ induces injective morphisms $M \xleftarrow{\tilde{f}} \tilde{X} \xrightarrow{\tilde{g}} N^1$ with the properties that $\tilde{f}$ is an $\epsilon$-equivalence and that $\tilde{g}$ is a $\delta$-equivalence. Hence $f$ and $g$ may be assumed to be injective. Under this assumption we get the following two inequalities from Lemma~\ref{lem:mono-1} 
\begin{align*}
d_I(M,X^{\eps/2})&\leq \eps/2\\
d_I(N^1,X^{\delta/2})&\leq \delta/2.
\end{align*}
Observe that $d_I(N^1,X^{\delta/2}) = d_I(N^{1-\delta/2},X) = d_I(N^{1+(\eps-\delta)/2},X^{\eps/2})$. Together with the first inequality this gives $d_I(M,N^{1+(\eps-\delta)/2}) \leq (\eps+\delta)/2$, and thus $$d_I(M,N) \leq d_I(M,N^{1+(\eps-\delta)/2}) + d_I(N^{1+(\eps-\delta)/2},N) \leq (\eps+\delta)/2 + (1+(\eps-\delta)/2) = 1+\eps.$$ 

To conclude the proof we will show that $\delta\geq 2$, as this implies $1+\eps < 1+4-\delta\leq 3$. Assuming that $n\geq 1$, let $p$ be such that $\dim N_p >0$ and $\dim M_{p+(r,r)} =0$ for all $r<1$. Such a point exists for the following reason: let $M_i$ be any indecomposable summand of $M$ and let $N_j$ be any indecomposable summand of $N$. Then $M_i$ is a staircase module for which the underlying staircase is obtained by moving certain corners of the staircase $S$ in \cref{fig:base_staircase}. Likewise, the staircase supporting $N_j$ is obtained by moving certain corners of $T$. However, by construction, and as shown in \cref{fig:staircase_shifted}, a number of corners are left unmoved. Hence, we may simply choose $p$ to be any corner point of $T$ with negative 1st coordinate which is left unmoved in the construction of $N_j$. 

Let $q= p-(1,1)$. Then $\dim N^1_q > 0$ and $\dim M_{q+(r,r)} = 0$ for all $r<2$. But since $f$ is an injection, the space $X_{q+(r,r)}$ must also be trivial for any $r<2$. It follows that $\delta \geq 2$. 
\end{proof}

\section{Conclusion}
\label{sec:conclusion}
Using the link between persistence modules indexed over $\R^2$ and CI problems introduced in \cite{bb-computational}, we settle the computational complexity of a series of problems. Most notably, we show that computing the interleaving distance is NP-hard, as is approximating it to any constant factor $<3$. Moreover, we investigated the problem of deciding one-sided stability. Except for checking isomorphism, which is known to be polynomial, we show that all non-trivial cases are NP-hard. This includes checking whether a module is a submodule of another. Our assumption that we are working over a finite field stays in the background for most of the paper, but we rely heavily on this assumption for proving the submodule problem. Lastly, we showed that approximating a distance $d$ arising from a noise system up to a constant less than $2$ is also NP-hard.

Throughout, we use persistence modules decomposing into very simple modules called staircase modules. These have the big advantage that the morphisms between them have very simple descriptions in terms of matrices. While this simplification might appear to throw the complexity out with the bathwater, our results clearly show that this is not the case.

The question of whether \textsc{$c$-Approx-Interleaving-Distance} is NP-hard for $c\geq 3$ is still open, and it is not clear whether one can prove this with CI problems or not. Even if this should not be possible, we believe that a better understanding of CI problems would lead to a better understanding of persistence modules and interleavings.

\paragraph{Acknowledgments.}
We thank the anonymous referees for valuable suggestions, 
including the connection to noise systems discussed in Section~\ref{sec:noise}.
Magnus Bakke Botnan has been partially supported by the DFG Collaborative Research Center SFB/TR 109 “Discretization in Geometry and Dynamics”. Michael Kerber is supported by Austrian Science Fund (FWF) grant number P 29984-N35.

\bibliography{bib}

\begin{thebibliography}{10}

\bibitem{bauer2014induced}
Ulrich Bauer and Michael Lesnick.
\newblock Induced matchings of barcodes and the algebraic stability of
  persistence.
\newblock In {\em Proceedings of the thirtieth annual symposium on
  Computational geometry}, page 355. ACM, 2014.

\bibitem{bmmps-persistent}
Paul Bendich, J~Steve Marron, Ezra Miller, Alex Pieloch, and Sean Skwerer.
\newblock Persistent homology analysis of brain artery trees.
\newblock {\em The annals of applied statistics}, 10 1:198--218, 2016.

\bibitem{bjerkevik-stability}
H{\aa}vard~Bakke Bjerkevik.
\newblock Stability of higher-dimensional interval decomposable persistence
  modules.
\newblock {\em CoRR}, abs/1609.02086, 2016.

\bibitem{bb-computational}
H{\aa}vard~Bakke Bjerkevik and Magnus~Bakke Botnan.
\newblock Computational complexity of the interleaving distance.
\newblock In {\em 34th International Symposium on Computational Geometry, SoCG
  2018, June 11-14, 2018, Budapest, Hungary}, pages 13:1--13:15, 2018.

\bibitem{bl-algebraic}
Magnus~Bakke Botnan and Michael Lesnick.
\newblock Algebraic stability of zigzag persistence modules.
\newblock {\em Algebraic \& Geometric Topology}, 18:3133--3204, 2018.

\bibitem{brooksbank2008testing}
Peter~A Brooksbank and Eugene~M Luks.
\newblock Testing isomorphism of modules.
\newblock {\em Journal of Algebra}, 320(11):4020--4029, 2008.

\bibitem{be-realizations}
Micka{\"{e}}l Buchet and Emerson~G. Escolar.
\newblock Realizations of indecomposable persistence modules of arbitrarily
  large dimension.
\newblock In {\em 34th International Symposium on Computational Geometry, SoCG
  2018, June 11-14, 2018, Budapest, Hungary}, pages 15:1--15:13, 2018.

\bibitem{chazal2016structure}
Fr{\'e}d{\'e}ric Chazal, Vin De~Silva, Marc Glisse, and Steve Oudot.
\newblock {\em The structure and stability of persistence modules}.
\newblock Springer, 2016.

\bibitem{cgos-persistence}
Fr{\'e}d{\'e}ric Chazal, Leonidas~J. Guibas, Steve~Y. Oudot, and Primoz Skraba.
\newblock Persistence-based clustering in riemannian manifolds.
\newblock {\em J. ACM}, 60(6):41:1--41:38, November 2013.

\bibitem{dx-computing}
Tamal~K. Dey and Cheng Xin.
\newblock Computing bottleneck distance for 2-d interval decomposable modules.
\newblock In {\em 34th International Symposium on Computational Geometry, SoCG
  2018, June 11-14, 2018, Budapest, Hungary}, pages 32:1--32:15, 2018.

\bibitem{eh-computational}
Herbert Edelsbrunner and John Harer.
\newblock {\em Computational Topology: An Introduction}.
\newblock American Mathematical Society, Providence, RI, USA, 2010.

\bibitem{ivanyos2010deterministic}
G{\'a}bor Ivanyos, Marek Karpinski, and Nitin Saxena.
\newblock Deterministic polynomial time algorithms for matrix completion
  problems.
\newblock {\em SIAM journal on computing}, 39(8):3736--3751, 2010.

\bibitem{kmn-geometry}
Michael Kerber, Dmitriy Morozov, and Arnur Nigmetov.
\newblock Geometry helps to compare persistence diagrams.
\newblock {\em {ACM} Journal of Experimental Algorithmics}, 22, 2017.

\bibitem{lesnick-theory}
Michael Lesnick.
\newblock The theory of the interleaving distance on multidimensional
  persistence modules.
\newblock {\em Foundations of Computational Mathematics}, 15(3):613--650, 2015.

\bibitem{lw-interactive}
Michael Lesnick and Matthew Wright.
\newblock Interactive visualization of 2-d persistence modules.
\newblock {\em CoRR}, abs/1512.00180, 2015.

\bibitem{pewvkjw-topology}
Pratyush Pranav, Herbert Edelsbrunner, Rien {van de Weygaert}, Gert Vegter,
  Michael Kerber, Bernard Jones, and Mathijs Wintraecken.
\newblock The topology of the cosmic web in terms of persistent betti numbers.
\newblock {\em Monthly Notices of the Royal Astronomical Society},
  465(4):4281--4310, 2016.

\bibitem{rhbk-stable}
Jan Reininghaus, Stefan Huber, Ulrich Bauer, and Roland Kwitt.
\newblock A stable multi-scale kernel for topological machine learning.
\newblock In {\em {IEEE} Conference on Computer Vision and Pattern Recognition,
  {CVPR} 2015, Boston, MA, USA, June 7-12, 2015}, pages 4741--4748, 2015.

\bibitem{rybakken2017decoding}
Erik Rybakken, Nils Baas, and Benjamin Dunn.
\newblock Decoding of neural data using cohomological feature extraction.
\newblock {\em Neural computation}, 31(1):68--93, 2019.

\bibitem{scolamiero2017multidimensional}
Martina Scolamiero, Wojciech Chach{\'o}lski, Anders Lundman, Ryan Ramanujam,
  and Sebastian {\"O}berg.
\newblock Multidimensional persistence and noise.
\newblock {\em Foundations of Computational Mathematics}, 17(6):1367--1406,
  2017.

\end{thebibliography}
\bibliographystyle{plain}

\end{document}